\documentclass[preprint,12pt]{elsarticle}
\usepackage{amsmath,amssymb}
\usepackage{amsmath,amssymb,amsfonts}
\usepackage{graphicx}
\usepackage{textcomp}
\usepackage[dvipsnames]{xcolor}
\def\BibTeX{{\rm B\kern-.05em{\sc i\kern-.025em b}\kern-.08em
    T\kern-.1667em\lower.7ex\hbox{E}\kern-.125emX}}
\definecolor{keywordcolor}{RGB}{10,100,227}

\usepackage{booktabs}%
\usepackage{array, longtable}
\newcolumntype{L}[1]{>{\raggedright\let\newline\\\arraybackslash\hspace{0pt}}p{#1}}
\usepackage{xspace}
\usepackage{paralist}
\usepackage{graphicx} % Required for inserting images
\usepackage{iftex}
\usepackage{programs}
\usepackage{url}
\usepackage{tablefootnote}
\usepackage{caption}
\usepackage{subcaption}
\definecolor{lightblue}{RGB}{173, 216, 230}
\newcommand{\hlblue}[1]{\colorbox{lightblue}{#1}}
\newcommand{\textcolorhlblue}[2]{\hlblue{\textcolor{#1}{#2}}}
\usepackage{adjustbox}
\usepackage{amsthm} %place before newtheorem
\newtheorem{thm}{Theorem}

\usepackage{cmap}
\usepackage[mathcal]{euscript}
\usepackage{listings}
\usepackage{amsfonts}
\usepackage{amsmath} 
\usepackage{graphicx}
\usepackage{proof}
\usepackage{url}
\usepackage{centernot}
\usepackage{courier}

\usepackage{algorithm}
\usepackage{algorithmicx}
\usepackage[noend]{algpseudocode}

\algnewcommand{\LeftComment}[1]{\Statex \(\triangleright\) #1}
\algnewcommand{\LeftCommentC}[1]{\Statex \(\mbox{\tt //}\) #1}

\newcommand{\inclusive}{\mbox{\tt inclusive}\xspace }
\newcommand{\exclusive}{\mbox{\tt exclusive}\xspace }
\newcommand{\required}{\mbox{\tt required}\xspace }

\newcommand{\always}{\mbox{\tt always}\xspace }
\newcommand{\eventually}{\mbox{\tt eventually}\xspace }
\newcommand{\until}{\mbox{\tt until}\xspace }

\newcommand{\once}{\mbox{\tt once}\xspace }

\newcommand{\timed}{\mbox{\tt timed}}

\newcommand{\before}{\mbox{\tt before}\xspace}

\newcommand{\after}{\mbox{\tt after}\xspace}
\renewcommand{\between}{\mbox{\tt between}\xspace}
\newcommand{\saltuntil}{\mbox{\tt until}\xspace}
\newcommand{\saltreleases}{\mbox{\tt releases}\xspace}
\newcommand{\saltand}{\mbox{\tt and}\xspace}
\newcommand{\saltor}{\mbox{\tt or}\xspace}
\newcommand{\saltnot}{\mbox{\tt not}\xspace}
\newcommand{\saltimplies}{\mbox{\tt implies}\xspace}

\newcommand{\saltalways}{\mbox{\tt always}\xspace}
\newcommand{\salteventually}{\mbox{\tt eventually}\xspace}
\newcommand{\saltnext}{\mbox{\tt next}\xspace}

\newcommand{\weak}{\mbox{\tt weak}\xspace}
\newcommand{\optional}{\mbox{\tt optional}\xspace}

\newcommand{\stopcond}{\textsc{stop}\xspace}

\newcommand{\SALT}{\textsc{SALT}}
\newcommand{\tFormula}{$\mathit{tFormula}$}
\newcommand{\tform}{$\mathit{tform}$}
\newcommand{\ltlFormula}{$\mathit{ltlFormula}$}
\newcommand{\ltlform}{$\mathit{ltlform}$}

\newcommand{\pFormula}{$\mathit{pFormula}$}
\newcommand{\baseForm}{$\mathit{baseForm}$}
\newcommand{\trigger}{$\mathit{trigger}$}

\newcommand{\formula}{$\mathit{formula}$}
\newcommand{\sform}{$\mathit{sform}$}
\newcommand{\LEFT}{\textsc{Left}}
\newcommand{\RIGHT}{\textsc{Right}}
\newcommand{\INSTANTIATE}{\textsc{Instantiate}}

\newcommand{\FFiM}{\mbox{\rm FFiM}\xspace}
\newcommand{\FLiM}{\mbox{\rm FLiM}\xspace}
\newcommand{\FiM}{\mbox{\rm FiM}\xspace}
\newcommand{\LiM}{\mbox{\rm LiM}\xspace}
\newcommand{\FNiM}{\mbox{\rm FNiM}\xspace}
\newcommand{\LNiM}{\mbox{\rm LNiM}\xspace}
\newcommand{\FTP}{\mbox{\rm FTP}\xspace}

\newcommand{\modevariable}{\textsc{mode}\xspace}  % mode variable

\newcommand{\bound}{\textsc{bound}\xspace}

 \newcommand{\formulaltlnospace}{formulaltl}
 \newcommand{\probForm}{probform\xspace}
\newcommand{\durn}{\textsc{n}}
\newcommand{\dur}{\textsc{n}\xspace}
\newcommand{\response}{\textsc{res}\xspace}
\newcommand{\saltap}{\textsc{probform}\xspace}

\newcommand{\condition}{\textsc{cond}\xspace}

\usepackage[noend]{algpseudocode}
\let\oldReturn\Return
\renewcommand{\Return}{\State\oldReturn}
\usepackage{diagbox}
\usepackage{xcolor}
 \usepackage{float}
\usepackage{mathtools}
\usepackage[framemethod=TikZ]{mdframed}
\usepackage[]{inputenc}
\usepackage[T1]{fontenc}
\usepackage{placeins}
\usepackage{xspace}
\usepackage{hyperref}
\usepackage{booktabs}
\usepackage{fancyvrb}
\usepackage{relsize}
\graphicspath{{images/}}
\usepackage{framed}
\usepackage{multirow}
\usepackage{tikz}
\usepackage{pifont}

\usepackage{amsmath}
\usetikzlibrary{arrows}
\newcommand{\fret}{\textsc{FRET}\xspace}
\newcommand{\fretish}{\textsc{FRETish}\xspace}
\newcommand{\salt}{\textsc{SALT}\xspace}

\newcommand{\prism}{\textsc{prism}\xspace}

\newcommand{\leftend}{\textsc{left}\xspace}   % scope interval left end
\newcommand{\leftendnospace}{\textsc{left}}

\newcommand{\commentout}[1]{}

\newcounter{template}

\newcommand{\fretafter}{{\textit{after}}\xspace}
\newcommand{\fretbefore}{{\textit{before}}\xspace}
\newcommand{\fretduring}{{\textit{during}}\xspace}
\newcommand{\fretin}{{\textit{in}}\xspace}
\newcommand{\fretnotin}{{\textit{not~in}}\xspace}

\newcommand{\fretonlyafter}{{\textit{only after}}\xspace}
\newcommand{\fretonlybefore}{{\textit{only before}}\xspace}
\newcommand{\fretonlyin}{{\textit{only in}}\xspace}

\newcommand{\fretimmediately}{{\textit{immediately}}\xspace}
\newcommand{\fretnext}{{\textit{next}}\xspace}
\newcommand{\fretnever}{{\textit{never}}\xspace}
\newcommand{\freteventually}{{\textit{eventually}}\xspace}
\newcommand{\fretalways}{{\textit{always}}\xspace}
\newcommand{\fretwithin}{{\textit{within}}\xspace}

\newcommand{\fretfor}{{\textit{for}}\xspace}
\newcommand{\fretuntil}{{\textit{until}}\xspace}

\newcommand{\mode}{\textsc{mode}\xspace}

\newcommand{\rightend}{\textsc{right}\xspace}  % scope interval right end
\newcommand{\rightendnospace}{\textsc{right}}  % scope interval right end

\definecolor{cdarkgreen}{rgb}{0.0,0.4,0.0}
\definecolor{customblue}{rgb}{0.0,0.0,0.7}
\definecolor{cbluegreen}{rgb}{0.0,0.4,0.7}
\definecolor{cpurple}{rgb}{0.5,0.0,0.7}
\definecolor{corange}{rgb}{0.8,0.6,0.2}
\definecolor{cgreen}{rgb}{0,0.6,0}
\colorlet{commentcolour}{green!50!black}
\colorlet{keywordcolour}{magenta!90!black}
\definecolor{MyDarkGreen}{rgb}{0.0,0.4,0.0}
\definecolor{MyBlue}{rgb}{0.0,0.0,0.7}
\definecolor{MyPurple}{rgb}{0.7,0.0,0.7}

\colorlet{punct}{red!60!black}
\definecolor{background}{HTML}{EEEEEE}
\definecolor{delim}{RGB}{20,105,176}
\colorlet{numb}{magenta!60!black}

\newcommand{\bArrow}[1]{\boldsymbol{#1}}

\lstdefinelanguage{smv}{
    literate=
      {->}{{{\color{corange}{$\bArrow{\rightarrow}$}}}}{1}
      {!}{{{\color{corange}{\textbf{!}}}}}{1}
      {\&}{{{\color{corange}{\textbf{\&}}}}}{1}
      {|}{{{\color{corange}{\textbf{|}}}}}{1}
      {=}{{{\color{corange}{\textbf{=}}}}}{1}
      {>}{{{\color{corange}{\textbf{>}}}}}{1}
      {<}{{{\color{corange}{\textbf{<}}}}}{1}
      {:}{{{\color{corange}{\textbf{:}}}}}{1}
      {)}{{{\color{corange}{\textbf{)}}}}}{1}
      {(}{{{\color{corange}{\textbf{(}}}}}{1}
      {?}{{{\color{corange}{\textbf{?}}}}}{1}
      {;}{{{\color{corange}{\textbf{;}}}}}{1}
      {+}{{{\color{corange}{\textbf{+}}}}}{1},
    morekeywords=[1]{LAST,FTP},
    keywordstyle=[1]\color{MyPurple}\bfseries,
    morekeywords=[2]{G, S, H, F, X, Y, U, O, SI},
    keywordstyle=[2]\color{customblue}\bfseries,
    morekeywords=[3]{MODULE,main,VAR,ASSIGN,DEFINE,LTLSPEC,NAME},
    keywordstyle=[3]\bfseries,
    morekeywords=[4]{LAST,MODE,COND,RES},
    keywordstyle=[4]\color{MyPurple}\bfseries,
    morekeywords=[5]{V},
    keywordstyle=[5]\color{corange}\bfseries,
    basicstyle=\scriptsize \ttfamily
}

\usepackage{todonotes}

\usepackage{tabularx}

\newcommand{\new}[1]{#1}

\usepackage{float}

\usepackage{fretish}

\definecolor{prismblue}{rgb}{0, 0, 0.8}
\definecolor{prismgreen}{rgb}{0, 0.6, 0}

\lstdefinelanguage{Prism}{ % syntax highlight via font
        basicstyle=\color{blue}\scriptsize\sffamily, % small true type font (like courier) //black
        keywords= {bool,C,ceil,const,ctmc,double,dtmc,endinit,endmodule,endrewards,endsystem,F,false,floor,formula,G,global,I,init,int,label,max,mdp,min,module,nondeterministic,P,Pmin,Pmax,prob,probabilistic,R,rate,rewards,Rmin,Rmax,S,stochastic,system,true,U,X,},
        keywordstyle={\bfseries\color{black}},
        numberstyle=\tiny\color{black},
        belowcaptionskip=\baselineskip,
        comment=[l] {//}, morecomment=[s]{/*}{*/}, % single and multi-line
        commentstyle= \color{prismgreen}, % dark green
        tabsize=4, % tab treatment (going to be fixed in Prism)
        captionpos=b, % put captions at the bottom
        escapechar=@, % write LaTeX comments escaped by @ symbol
        literate={;}{\textcolor{black}{;}}1 % Coloring ';', '=',... in black
                 {=}{\textcolor{black}{=}}1
                 {:}{\textcolor{black}{:}}1
                 {[}{\textcolor{black}{[}}1
                 {]}{\textcolor{black}{]}}1
                 {<}{\textcolor{black}{$<$}}1
                 {>}{\textcolor{black}{$>$}}1 
                 {+}{\textcolor{black}{$+$}}1 
                 {->}{\textcolor{black}{$\rightarrow$}}1
                 {&}{\textcolor{black}{\&}}1 
                 {"}{\textcolor{black}{``}}1
}

\begin{document}
\begin{frontmatter}
\title{Automated Formalization of Probabilistic Requirements from Structured Natural Language}
%\title{Automating the Formalization of Probabilistic Requirements}
  % One or more authors
  \author[inst1]{Anastasia Mavridou\corref{cor1}}
  \ead{anastasia.mavridou@nasa.gov}
  \cortext[cor1]{Corresponding author}
    \author[inst2]{Marie Farrell\corref{cor1}}
  \ead{marie.farrell@manchester.ac.uk}
%  \cortext[cor1]{Corresponding author}
    \author[inst3]{Gricel V\'azquez\corref{cor1}}
  \ead{gricel.vazquez@york.ac.uk}
%  \cortext[cor1]{Corresponding author}
      \author[inst4]{Tom Pressburger}
  \ead{tom.pressburger@nasa.gov}
%  \cortext[cor1]{Corresponding author}
      \author[inst5]{Timothy E. Wang}
  \ead{timothy.wang@rtx.com}
%  \cortext[cor1]{Corresponding author}
      \author[inst3]{Radu Calinescu}
  \ead{radu.calinescu@york.ac.uk}
 % \cortext[cor1]{Corresponding author}
      \author[inst2]{Michael Fisher}
  \ead{michael.fisher@manchester.ac.uk}
 % \cortext[cor1]{Corresponding author}

  % Affiliation(s)
  \affiliation[inst1]{KBR Inc., NASA Ames Research Center
    addressline={Moffett Field}, city={California}, postcode={CA 94035}, country={USA}}
  \affiliation[inst2]{organization={Department of Computer Science, The University of Manchester},
    addressline={Oxford Road}, city={Manchester}, postcode={M13 9PL}, country={United Kingdom}}
  \affiliation[inst3]{organization={Department of Computer Science, University of York},
    addressline={Heslington}, city={York}, postcode={YO10 5DD}, country={United Kingdom}}
  \affiliation[inst4]{organization={NASA Ames Research Center},
    addressline={Moffett Field}, city={California}, postcode={CA 94035}, country={USA}}
  \affiliation[inst5]{organization={RTX Technology Research Center},
    addressline={411 Silver Ln}, city={East Hartford}, postcode={CT 06118}, country={USA}}
%===============================================================
% Abstract
%===============================================================
\begin{abstract}
\new{
\noindent \textbf{Context:} Integrating autonomous and adaptive 
behavior into software-intensive systems presents significant challenges for software development, as uncertainties in the environment or decision-making processes must be explicitly captured. These challenges are amplified in safety- and mission-critical systems, which must undergo rigorous scrutiny during design and development. Key among these challenges is the difficulty of specifying requirements that use probabilistic constructs to capture the uncertainty affecting these systems. To enable formal analysis, such requirements must be expressed in precise mathematical notations such as probabilistic logics. However, expecting developers to write requirements directly in complex formalisms is unrealistic and highly error-prone. \\
\noindent \textbf{Objectives:} We extend the structured natural language used by NASA's Formal Requirement Elicitation Tool (\fret) with support for the specification of unambiguous and correct probabilistic requirements, and develop an automated approach for translating these requirements into logical formulas. \\
\noindent \textbf{Methods:} We propose and develop a formal, compositional, and automated approach for translating structured natural-language requirements into formulas in probabilistic temporal logic. To increase trust in our formalizations, we provide assurance that the generated formulas are well-formed and conform to the intended semantics through an automated validation framework and a formal proof. \\
\noindent \textbf{Results:} We demonstrate applicability through a comprehensive evaluation using more than $300$ requirements from the research literature, public repositories and an industry-led case study. The results of this evaluation show that our approach can successfully specify a wide range of probabilistic requirements, confirming its expressiveness and utility. \\
\noindent \textbf{Conclusion: } The extended FRET tool enables developers to specify probabilistic requirements in structured natural language, and to automatically translate them into probabilistic temporal logic, making the formal analysis of autonomous and adaptive systems more practical and less error-prone.}
\end{abstract}

%===============================================================
% Keywords
%===============================================================
\end{frontmatter}

%===============================================================
% Sections
%===============================================================
\section{Introduction}
\label{sec:intro}
In hazardous environments, such as nuclear decommissioning and space exploration, there is an increased desire and need for autonomous systems~\cite{fisher2021overview}. These systems must react to an, often uncertain, environment whilst relying on imperfect and degraded sensor data. Autonomous systems in these domains are typically (mission- and/or safety-) critical, as failures can lead to the loss of costly operations and/or pose serious risks to human safety. For critical systems, various standards and regulations, such as DO-178C in the aerospace domain, mandate a requirements-driven approach to development~\cite{rierson2017developing}. \new{The fact that autonomy is often achieved via the use of learning-enabled components (LECs), }combined with environmental variability and imperfect sensing, demands requirements that employ probabilities to effectively capture uncertainty~\cite{farrell2023exploring}. 

There are two important challenges associated with Requirements Engineering (RE), both aggravated by the use of probabilities in autonomous systems. First, RE is an inherently human-centric activity and so the language chosen to specify requirements %and how intuitive this is, 
is important for usability and interpretation of requirements. To this end, developers typically write requirements in intuitive natural language, which however can be extremely ambiguous~\cite{farrell2022fretting}. \new{Formal languages such as logic-based ones avoid such ambiguity, but expecting developers to write requirements directly in these complex formalisms is unrealistic and highly error-prone~\new{\cite{rozier2016specification}}.} Adding probabilities  exacerbates the existing ambiguity issue. The second challenge is related to connecting requirements with formal analysis tools to verify whether the systems under development meet their requirements. 
% and error-prone, requiring domain expects to write requirements directly in these complex formalisms is often impractical and highly error-prone~\new{\cite{rozier2016specification}}.
A variety of analysis techniques have been developed for specifications written in probabilistic temporal logics, including probabilistic model checking and runtime monitoring~\cite{KNP11,dehnert2017storm,hahn2014iscas,grunske09}. However, expecting developers to directly specify requirements in such complex formalisms is often error prone or even outside their skill set.

To address these challenges, NASA Ames Research Center with contributions from multiple other research teams, developed the open-source tool \fret~\cite{giannakopoulou2020formal}, which enables developers to write unambiguous requirements in a structured natural language called \fretish. The underlying semantics of a \fretish requirement is defined by the combination of values assigned to three key fields of the language: \textit{scope}, \textit{condition} and \textit{timing}. Each unique combination of these field values corresponds to a specific \emph{template key}, which acts as a blueprint that guides how the requirement is translated into a formal representation. By allowing $8$ options for \textit{scope}, $2$ options for 
\textit{condition} and $10$ options for \textit{timing}, \fret supports $8 \times 2 \times 10 = 160$
distinct template keys. Even though \emph{classic \fret} is already expressive, none of these combinations supports the specification of probabilistic requirements.

Specifically, we contribute the following:
\begin{enumerate}
    \item Extensions to the \fretish language, including
(i)  a new \textit{probability} field for expressing \emph{probabilistic requirements}; and
(ii) a new value for the \textit{condition} field. %; and
%(iii) additional constructs offering alternative phrasings for existing semantics.
%While the third extension does not increase expressiveness, it enhances usability by allowing multiple ways to express the same concepts.
%
These two extensions add 80 new template keys to the classic \fretish{} and 320 to its probabilistic extension, significantly increasing the total from 160 to 560.
    \item A formal, compositional and automated translation from the extended \fretish language into Probabilistic Computation Tree Logic Star (PCTL$^*$)~\cite{aziz1995usually,baieralgorithmic} formulas. 
    We prove that the generated  formulas are well-formed.
    \item An oracle-based, automated validation framework that helps to ensure that the generated PCTL$^*$  formulas conform to the intended semantics of the requirements. 
    \item An extensive evaluation of the effectiveness of our approach on a diverse set of requirements --- both formal and natural language --- sourced from the research literature and industry.
\end{enumerate}
Our tool and evaluation requirements are all available at~\cite{FRETGithubX}.

\subsection{Related work}  
The elicitation, formalization, and analysis of requirements has been the focus of significant research by the automated software engineering community (e.g.,~\cite{feng2023towards,sorathiya2024towards,fazelnia2024translation,kolthoff2024self}). Many approaches aim to simplify formal specification by defining reusable patterns that can then be instantiated to obtain system requirements. A notable example is the Specification Pattern System (SPS)~\cite{dwyer1999patterns} with later extensions for real-time properties~\cite{konrad2005real}. Tools like Prospec \cite{Prospec}, SPIDER \cite{SPIDER} and SpeAR \cite{SpeAR} assist users in writing pattern-based requirements. The EARS approach~\cite{EARS2} offers five informal templates that have been shown to cover most high-level requirements. Other tools such as STIMULUS~\cite{STIMULUS} ASSERT~\cite{ASSERT,ASSERT2}, and SLEEC~\cite{getir2023specification} use restricted natural language (e.g., drag and drop phrases) to support specification. However, none of these approaches and tools support probabilistic requirements. 

%The SLEEC tool provides a domain-specific language for non-functional requirements with its semantics written in tock-CSP~\cite{getir2023specification}, which can be particularly useful for autonomous system specifications.%, and their redundancy and conflict resolution through verification engines. %, such as the FDR model checker. 
 %However, currently it does allow users to express uncertain behavior. %Adding probabilistic constructs is cited by the authors as future work. 

Other approaches focus on specifying patterns for probabilistic properties. ProProST \cite{grunske2008specification} defines eight patterns that were derived from the literature and formalized
in probabilistic logics. Like \fret, ProProST supports specification in a structured English grammar to express requirements naturally. The Property Specification Pattern Framework (PSPFramework) \cite{autili2015aligning}, unifies and extends previous work on real-time and probabilistic property patterns. The PSPWizard tool supports building formalizations by choosing from its pattern library. Similarly, the QUARTET catalog provides RPCTL (PCTL augmented with rewards~\cite{forejt2011automated}) patterns for the specification of precise robotic mission requirements~\cite{QUARTET,vazquez2024robotics},  supported by a tool that provides a pattern-based domain-specific language. 

Unlike other approaches that rely on menu-based selection from predefined pattern libraries, \fret enables users to express requirements directly in a restricted form of natural language. This bridges the gap between the rigor of formal specifications (often unintuitive) and the readability of unrestricted natural language (typically ambiguous).
Earlier classic \fret supported 160 template keys; our extensions increase this to 560 and add support for expressing uncertainty. Each template key is translated into formulas compositionally, avoiding manual, hardcoded mappings and thus enhancing maintainability. %For each template key, formulas are generated compositionally, eliminating the need for custom, manual translations. This compositional approach enhances maintainability, which would otherwise be constrained by hardcoded mappings between templates and formulas. 
Finally, we provide a validation framework that systematically checks the correctness of generated formulas, providing a level of assurance not offered by prior works.

\new{The remainder of this paper is structured as follows. We provide an overview of our approach in Section \ref{sec:overview}. Then, Section \ref{sec:specification} outlines the improvements that we have made to the \fretish specification language to enable the formalization of probabilistic requirements. Section \ref{sec:compositional_formalization} details our compositional algorithm for generating PCTL* specifications from probabilistic \fretish requirements. Section \ref{sec:testing} describes our automated validation framework, designed to check that the formalizations produced by FRET are correct. We provide a detailed evaluation in Section \ref{sec:evaluation}, where we assemble a corpus of 334~requirements and answer three research questions related to the expressibility and usefulness of our approach, along with assessing the commonalities amongst requirements in this corpus. In Section~\ref{sec:discussion}, we discuss several challenges that we encountered during this project, and the use of structured natural language as a mediator between LLM-generated natural language and formalized requirements. Section \ref{sec:threats} outlines the internal and external threats to validity that are applicable to this work. Finally, Section~\ref{sec:conclude} concludes the paper with a brief summary, and outlines future research directions.}
\section{Overview}
\label{sec:overview}
\begin{figure}[t]
    \centering
    \includegraphics[width=0.85\linewidth]{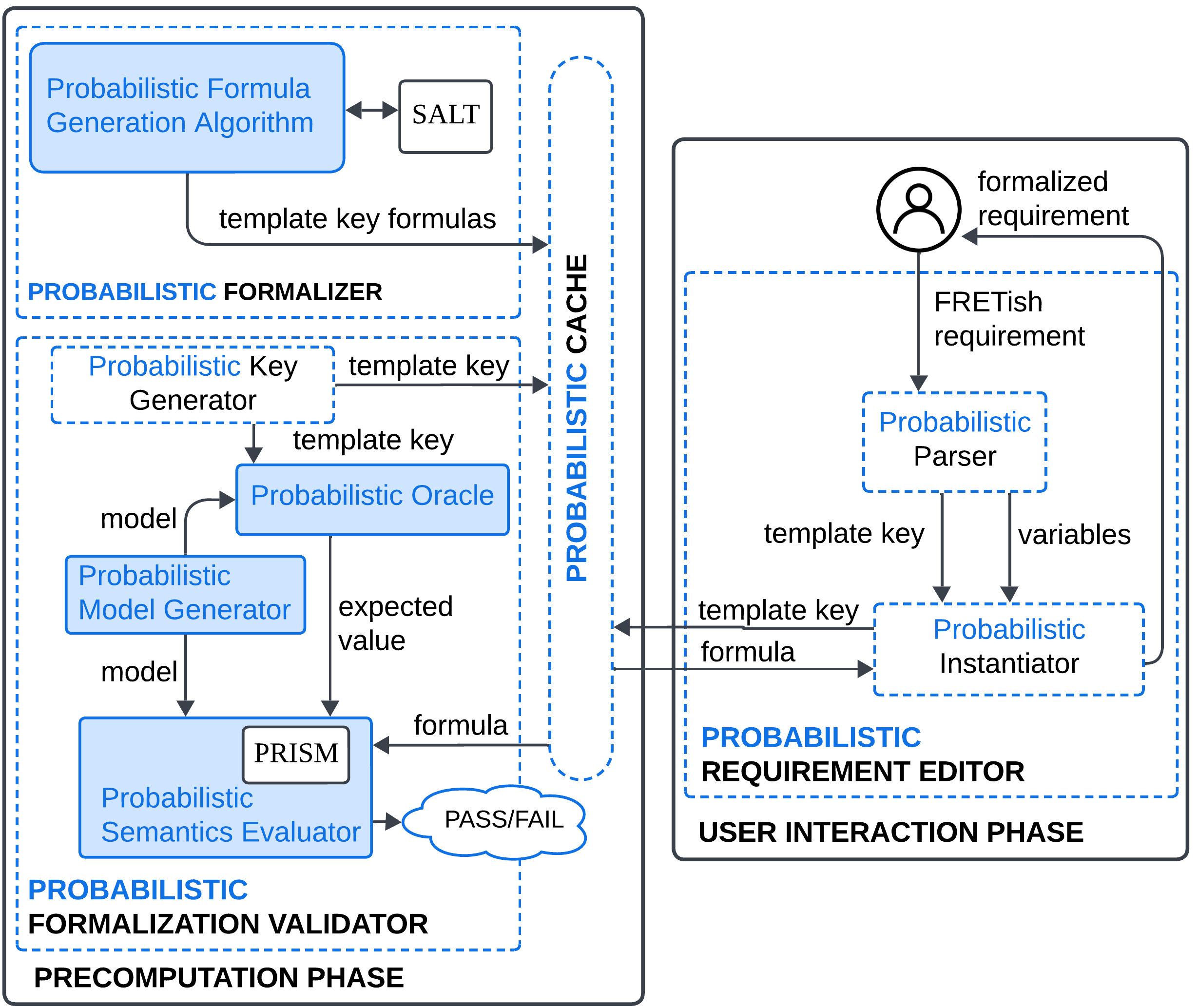}
    \caption{Probabilistic \fret: extensions to classic \fret components are shown with dashed lines; new components are highlighted with solid lines and a colored background.} %(e.g., %\textsf{\small{Probabilistic Oracle}}).}
    \label{fig:architecture}
\end{figure}

Figure~\ref{fig:architecture} presents the components of our framework. 
The figure is divided into two main parts: (i) the \textsf{\small{PRECOMPUTATION PHASE}} on the left hand side and (ii) the \textsf{\small{USER INTERACTION PHASE}} on the right hand side. The formalization process is performed once during the precomputation phase by the \textsf{\small{PROBABILISTIC FORMALIZER}} and then stored in the \textsf{\small{PROBABILISTIC CACHE}} in the form of a JSON file. Our \textsf{\small{Probabilistic Formula Generation Algorithm}} --- described in detail in Section~\ref{sec:algorithm} --- automatically translates each template key into a PCTL$^*$  formula. To support the generation of PCTL$^*$  \textit{path formulas}, which may embed arbitrary LTL expressions as their temporal core, we use the Structured Assertion Language for Temporal Logic (\salt)~\cite{SALT}. \salt supports the automated simplification and generation of Linear Temporal Logic (LTL)~\cite{manna1992temporal} from a higher-level specification language. The resulting PCTL$^*$  formulas  are expressed using the format of the PRISM property specification language~\cite{PrismPropSpec}. This choice of the PRISM language is motivated by its widespread support among probabilistic analysis tools, including the PRISM~\cite{KNP11}, STORM \cite{dehnert2017storm} and iscasMC \cite{hahn2014iscas} model checkers. Moreover, it serves as a convenient intermediate representation that can be easily adapted for compatibility with other tools, e.g.,  %tools for runtime verification of PCTL/PCTL$^*$  properties
~\cite{%forejt2012incremental, 
farrell2025quantitative}.

To validate correctness, the \textsf{\small{PROBABILISTIC FORMALIZATION VALIDATOR}} (Section~\ref{sec:testing}) offers a modular framework for checking that the generated formulas accurately capture the intended semantics. It includes the extended \textsf{\small{Probabilistic Key Generator}}, which produces all new template keys, and three new components: (i) the \textsf{\small{Probabilistic Model Generator}}, which constructs random models; (ii) the \textsf{\small{Probabilistic Oracle}}, which interprets the template key and returns the expected result for a given model; and (iii) the \textsf{\small{Probabilistic Semantics Evaluator}}, which uses PRISM to verify that the PCTL$^*$ formula, when evaluated on the model, matches the oracle’s output.

\begin{figure}[t]
    \centering
    %\setlength{\fboxrule}{0.4pt} % Adjust for desired thinness (e.g., 0.4pt, 0.2pt)
    %\setlength{\fboxsep}{1pt}   % Adjust for desired padding around the image
    %\fcolorbox{gray}{white}{ % Frame color is gray, background is white
      \includegraphics[width=\linewidth]{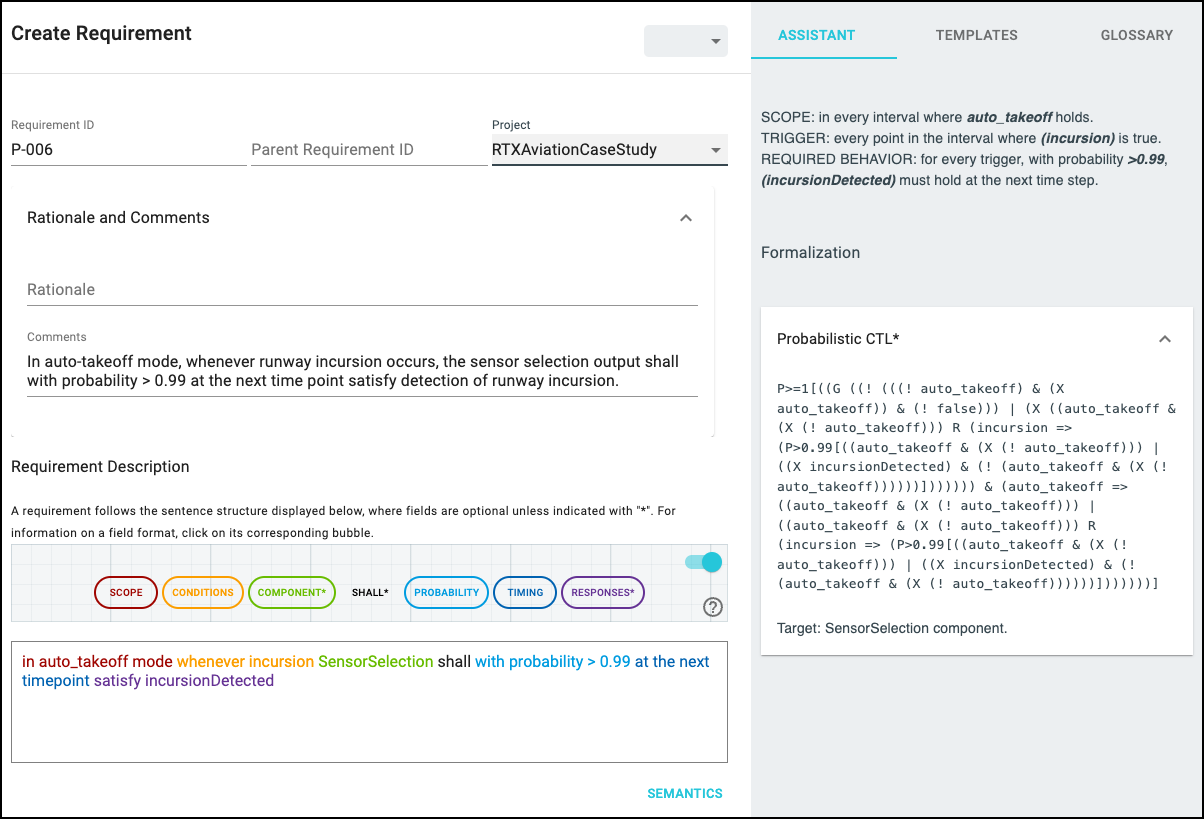}
    %}
    \caption{\new{\fret's probabilistic requirement editor. On the left, the requirement is entered in \fretish. The initial, unstructured natural language version of the requirement (pure English) is entered in the Comments field. On the right, \fret provides an explanation of the requirement semantics and the generated PCTL* formula.}}
    \label{fig:specification}
\end{figure}

On the right of Figure~\ref{fig:architecture} we see how users interact with \fret.  Requirements are written in extended \fretish (Section~\ref{sec:specification}). 

\new{Figure~\ref{fig:specification} shows the extended \fret editor with requirement \textbf{[P-006]} being entered. This requirement originates from an RTX\footnote{RTX is a corporate conglomerate of Collins Aerospace, Pratt \& Whitney and Raytheon.} Technology Research Center (RTRC) case study focused on autonomous taxiing, take-off, and landing systems (Section~\ref{sec:industrialcs}), and serves as a running example throughout this paper. The natural language version of the requirement is entered in the Comments section of the editor. As the user types, the \textsf{\small{Probabilistic Parser}} identifies the corresponding template key and extracts the relevant variables. The \textsf{\small{Probabilistic Instantiator}} then retrieves the corresponding PCTL* formula from the \textsf{\small{PROBABILISTIC CACHE}}, instantiates it with the extracted variables, and returns the formalized requirement to the user, as shown in the \emph{assistant tab} on the right side of Figure~\ref{fig:specification}.}

% Throughout the paper, we use  requirement \textbf{[P-006]} as a running example. This requirement originates from an RTX\footnote{RTX is a corporate conglomerate of Collins, Pratt \& Whitney and Raytheon.} industrial case study in the aviation domain, which focuses on autonomous capabilities such as autonomous taxiing, take-off and landing (see Section~\ref{section:casestudies} for a detailed description). The requirement is stated in natural language as follows:\\ \textit{``In auto-takeoff mode, whenever runway incursion occurs, the sensor selection output shall with probability > 0.99 at the next time point satisfy detection of runway incursion.''}.

% The high-level architecture of \fret's probabilistic formalization framework is illustrated in Figure~\ref{fig:architecture}. Formalization is driven by semantic \textit{template keys}, which are determined by the values assigned to the fields of each \fretish requirement. For instance, the \fretish representation of requirement \textbf{[P-006]} is:
%  \scope{in auto\_takeoff mode} \conditionF{whenever incursion} \component{SensorSelection} shall \probability{with probability $> 0.99$} \timing{at the next timepoint} \responseF{incursionDetected} 

\begin{table}[htbp]
\centering
\caption{Extended FRETish Grammar. \texttt{\small{ID}} is an alphanumerical string, \texttt{\small{prob\_bound}} \small{$\in \mathbb{R}^{[0,1]}$}  and \texttt{\small{NUMBER}} \small{$\in \mathbb{R}^{+}$}. \new{Grammar extensions are highlighted in \textcolorhlblue{blue}{blue}, while classic \fret constructs appear in black colored font.}}

{\renewcommand{\arraystretch}{0.99}
\resizebox{\textwidth}{!}{%
\begin{tabular}{l l p{16cm}}
\toprule
\textbf{Rule} & & \textbf{Definition} \\
\midrule
\texttt{requirement} & \texttt{::=} & \texttt{(scope)? (condition)? (component \textbf{SHALL} $\mid$ \textbf{SHALL} component) \textcolorhlblue{blue}{(probability)?} (timing)? response} \newline \\ 

\texttt{scope} & \texttt{::=} &
\texttt{\textbf{ONLY} ((\textcolorhlblue{blue}{\textbf{DURING} $\mid$} (\textcolorhlblue{blue}{(\textbf{WHEN} $\mid$ \textbf{IF})?} \textbf{IN})) scope\_mode} 
\texttt{ \textcolorhlblue{blue}{$\mid$ 
 \textbf{WHILE} scope\_condition}} \newline
\makebox[10pt][l]{} \texttt{ $\mid$ (\textbf{AFTER} $\mid$ \textbf{BEFORE}) (scope\_mode \textcolorhlblue{blue}{$\mid$ scope\_condition}))} \newline
\texttt{$\mid$ \textbf{EXCEPT} ((\textcolorhlblue{blue}{\textbf{DURING} $\mid$} (\textcolorhlblue{blue}{(\textbf{WHEN} $\mid$ \textbf{IF})?} \textbf{IN})) scope\_mode} %\newline
%\makebox[10pt][l]{}  
\texttt{ \textcolorhlblue{blue}{$\mid$ \textbf{WHILE} scope\_condition}})\newline
\texttt{$\mid$ \textcolorhlblue{blue}{(\textbf{WHEN} $\mid$ \textbf{IF})} \textbf{NOT}? \textbf{IN} scope\_mode} \newline
\texttt{$\mid$ (\textbf{IN} \textcolorhlblue{blue}{$\mid$ \textbf{DURING}}) scope\_mode} \newline
\textcolorhlblue{blue}{\texttt{$\mid$ \textbf{UNLESS IN} scope\_mode}} \newline
\textcolorhlblue{blue}{ \texttt{$\mid$ \textbf{WHILE} scope\_condition}} \newline
\texttt{$\mid$ (\textbf{AFTER} $\mid$ \textbf{BEFORE}) (scope\_mode $\mid$ scope\_condition)} \\
\texttt{scope\_mode} & \texttt{::=} & \texttt{(\textbf{MODE} mode\_name) $\mid$ (mode\_name \textbf{MODE}) $\mid$ mode\_name} \\
\texttt{mode\_name} & \texttt{::=} & \texttt{ID} \\
\textcolorhlblue{blue}{\texttt{scope\_condition}} & \textcolorhlblue{blue}{\texttt{::=}} &  \textcolorhlblue{blue}{\texttt{boolean\_expression}} \newline \\

\texttt{condition} & \texttt{::=} & \texttt{(\textbf{AND})? condition\_expression} \\
\texttt{condition\_expression} & \texttt{::=} & \texttt{qualified\_condition1 ( qualified\_condition2 )*} \\
\texttt{qualified\_condition1} & \texttt{::=} & \texttt{qualifier\_word pre\_condition \textcolorhlblue{blue}{(\textbf{IS} (\textbf{TRUE} $\mid$ \textbf{FALSE}))?}} \\
\texttt{qualified\_condition2} & \texttt{::=} & \texttt{(\textbf{AND} $\mid$ \textbf{OR})? qualifier\_word pre\_condition \textcolorhlblue{blue}{(\textbf{IS} (\textbf{TRUE} $\mid$ \textbf{FALSE}))?}} \\
\texttt{qualifier\_word} & \texttt{::=} & \texttt{\textbf{UPON} \textcolorhlblue{blue}{$\mid$ \textbf{WHENEVER}} $\mid$ \textbf{WHEN} $\mid$ \textbf{UNLESS} $\mid$ \textbf{WHERE} $\mid$ \textbf{IF}} \\
\texttt{pre\_condition} & \texttt{::=} & \texttt{boolean\_expression} \newline \\

\texttt{component} & \texttt{::=} & \texttt{(\textbf{THE})? component\_name} \\
\texttt{component\_name} & \texttt{::=} & \texttt{ID} \newline \\

\textcolorhlblue{blue}{\texttt{probability}} & \textcolorhlblue{blue}{\texttt{::=}} & \textcolorhlblue{blue}{\texttt{\textbf{WITH} probability\_aux}} \\
\textcolorhlblue{blue}{\texttt{probability\_aux}} & \textcolorhlblue{blue}{\texttt{::=}} & \textcolorhlblue{blue}{\texttt{\textbf{PROBABILITY} relational\_op prob\_bound}} \\
\textcolorhlblue{blue}{\texttt{relational\_op}} & \textcolorhlblue{blue}{\texttt{::=}} & 
\textcolorhlblue{blue}{\texttt{\textbf{<} $\mid$ \textbf{<=} $\mid$ \textbf{>} $\mid$ \textbf{>=}}} \newline \\

\texttt{timing} & \texttt{::=} &
\texttt{\textbf{WITHIN} duration $\mid$ \textbf{FOR} duration $\mid$ \textbf{AFTER} duration} \newline
\texttt{$\mid$ \textbf{UNTIL} stop\_condition $\mid$ \textbf{BEFORE} stop\_condition} \newline
\texttt{$\mid$ \textbf{AT} \textbf{THE} \textcolorhlblue{blue}{(\textbf{FIRST} $\mid$ \textbf{SAME} } $\mid$ }\texttt{ \textbf{NEXT}) \textbf{TIMEPOINT}} \newline
\texttt{$\mid$ \textbf{IMMEDIATELY} \textcolorhlblue{blue}{$\mid$ \textbf{INITIALLY}} $\mid$  \textbf{EVENTUALLY} $\mid$ \textbf{ALWAYS} $\mid$ \textbf{NEVER}} \\
\texttt{stop\_condition} & \texttt{::=} & \texttt{boolean\_expression} \\
\texttt{duration} & \texttt{::=} & \texttt{NUMBER timeunit} \\
\texttt{timeunit} & \texttt{::=} & \texttt{\textbf{TICKS} $\mid$ \textbf{MICROSECONDS} $\mid$ \textbf{MILLISECONDS} $\mid$ \textbf{SECONDS} $\mid$ \textbf{MINUTES} $\mid$ \textbf{HOURS}} \newline \\

\texttt{response} & \texttt{::=} & \texttt{satisfaction} \\
\texttt{satisfaction} & \texttt{::=} & \texttt{\textbf{SATISFY} post\_condition} \\
\texttt{post\_condition} & \texttt{::=} & \texttt{boolean\_expression} \\
\bottomrule
\end{tabular}
}
}
\label{table:fret-full-grammar}
\end{table}

\section{Requirement Specification Language}
\label{sec:specification}

In this section, we present our extensions to the classic FRETish language. For completeness, since our formalization algorithm (Section~\ref{sec:algorithm}) uses the full \fretish syntax, we describe both the classic FRET constructs~\cite{giannakopoulou2020formal, GIANNAKOPOULOU2021106590} and our enhancements, explicitly distinguishing where appropriate. Table~\ref{table:fret-full-grammar} presents the syntax of the extended FRETish grammar --- grammar extensions are highlighted in \textcolorhlblue{blue}{blue}, while classic \fret constructs appear in black colored font. The grammar's format follows ANTLR's grammar syntax~\cite{antlr4}, which is based on Extended Backus-Naur Form (EBNF). In the grammar, terminal symbols are given in capital-letters and \textbf{bold} font, whereas non-terminal symbols are given in a \texttt{monospaced} font. Optional elements are denoted using \mbox{`( )?'}. For example, \texttt{\small{(\textbf{WHEN)? IN}}} indicates that \texttt{\small{\textbf{WHEN}}} may be omitted if doing so results in a sentence that better resembles natural language for the user. The notation `()*' indicates zero or more repetitions of the grouped element, and `|' denotes alternative options.

 With our extensions, a \fretish requirement consists of seven fields. As the requirement is entered, the \fret editor dynamically colors the text corresponding to each field (see Figure~\ref{fig:specification}). The fields are \textcolor{Mahogany}{\small{\texttt{scope}}}; \textcolor{orange}{\small{\texttt{condition}}}; \textcolor{ForestGreen}{\small{\texttt{component*}}}; \texttt{\small{shall*}}; \textcolor{cyan}{\small{\texttt{probability}}}; \textcolor{blue}{\small{\texttt{timing}}}; and \textcolor{violet}{\small{\texttt{response*}}}. Fields marked with an asterisk are mandatory. 

We use requirement \textbf{[P-006]} to explain the \fretish fields. We start with the mandatory fields. The \textcolor{ForestGreen}{\small{\texttt{component}}} field specifies the system component that the requirement applies to (e.g., \textcolor{ForestGreen}{SensorSelection}), and is defined as a string. The \texttt{shall} keyword states that the component's behavior must conform to the requirement. The \textcolor{violet}{\small{\texttt{response}}} field has the form \emph{satisfy R} (e.g., \textcolor{violet}{satisfy incursionDetected}), where \emph{R} is a Boolean expression that the component's behavior must satisfy.

%Component behavior often depends on operational modes. 
The \textcolor{Mahogany}{\small{\texttt{scope}}} field specifies the \emph{modes of operation} relevant to the component's behavior (e.g., \textcolor{Mahogany}{in auto\_takeoff mode}) by defining the time intervals during which the requirement is enforced. If omitted, scope means \emph{global}, i.e., entire execution. For a \mode (e.g., \textcolor{Mahogany}{auto\_takeoff}), \fret supports seven relationships: \fretbefore\ \mode (the requirement is enforced strictly before the first point at which \mode holds); \fretafter\ \mode (the requirement is enforced strictly after the last point at which \mode holds); \fretin\ \mode (or the synonyms \fretduring, \emph{when in}, \emph{if in}, \emph{while}) the requirement is enforced when the component is in \mode; and 
\fretnotin\ \mode (or synonyms \emph{except in}, \emph{unless in}) the requirement is enforced when the component is \emph{not} in \mode. 
Sometimes, it is necessary to specify that a requirement is enforced \emph{only} within a particular time interval, meaning that it should \emph{not} be satisfied outside of that interval. For this, the scopes \fretonlyafter, \fretonlybefore, and \fretonlyin are provided.

The \textcolor{orange}{\small{\texttt{condition}}} field is a Boolean expression that triggers the response within the specified scope. Our extension supports two types: (i) a \emph{regular condition}, through the \emph{upon} keyword (or synonyms \emph{when}, \emph{where} and \emph{if}), which triggers the response to occur at the time that the condition becomes true (from false); (ii) a new \emph{holding condition} through the \emph{whenever} keyword (e.g., \textcolor{orange}{whenever incursion}), which triggers the response to occur at all time points when the condition is true.

The new \textcolor{cyan}{\small{\texttt{probability}}} field that we introduce in this work specifies a lower or upper probabilistic bound on the timing response, indicating the likelihood that the system meets the timing constraint on the response within the specified limits (e.g., \textcolor{cyan}{with probability > 0.99}).

The \textcolor{blue}{\small{\texttt{timing}}} field specifies when the response is expected to hold (e.g., \textcolor{blue}{at the next timepoint}) relative to the \textcolor{Mahogany}{\small{\texttt{scope}}} and \textcolor{orange}{\small{\texttt{condition}}}. \fret supports ten timing options: \fretimmediately (with synonyms \emph{initially}, \emph{at the same timepoint}, \emph{at the first timepoint}), \fretnext, \fretnever, \freteventually, \fretalways, \fretwithin~\dur
time units, \fretfor~\dur time units, \fretafter~\dur time units (i.e., not within \emph{\dur} time units and at the \durn+1$^{st}$ time unit),  \fretuntil (response holds until a stop condition), and \fretbefore (response holds before a stop condition). When timing is omitted, \freteventually is used as default.

In summary, key extensions of this work include the new \textcolor{cyan}{\small{\texttt{probability}}} field, which enables the specification of probabilistic requirements, and the introduction of \emph{holding} conditions. The latter significantly increases the expressiveness of classic \fretish (contributing 80 new template keys) and allows us to express infinitely-often (G F) requirements, which were not previously supported in \fret. It also enables the specification of certain probabilistic patterns, such as probabilistic response~\cite{grunske2008specification}.
 To support more natural expression of requirements, we also enriched the \fretish grammar with multiple syntactic alternatives for semantically equivalent constructs in the \textcolor{Mahogany}{\small{\texttt{scope}}}, \textcolor{orange}{\small{\texttt{condition}}} and \textcolor{blue}{\small{\texttt{timing}}} fields. E.g., the \emph{in} \mode construct can now be specified as \texttt{\small{\textbf{DURING} scope}}, \texttt{\small{\textbf{WHEN IN} scope}} or \texttt{\small{\textbf{WHILE} scope\_condition}}.

\subsection{Placement of the Probability Field.}
\label{sec:probabilityField}
The placement of the \textcolor{cyan}{\small{\texttt{probability}}} field was informed by an analysis of probabilistic requirements from industrial case studies and academic literature. We considered several candidate positions, including placing it before the mandatory \textcolor{violet}{\small{\texttt{response}}} field and at the start of the requirement, before the \textcolor{Mahogany}{\small{\texttt{scope}}} field. 

Placing the \textcolor{cyan}{\small{\texttt{probability}}} field before the \textcolor{violet}{\small{\texttt{response}}} field results in a Boolean expression wrapped by the PCTL$^*$ $P$ operator, yielding a formula that evaluates to either 0 or 1. This construction limits expressiveness and fails to support the specification of meaningful probabilistic constraints.

Positioning the \textcolor{cyan}{\small{\texttt{probability}}} field at the very beginning of the \fretish structure, prior to \textcolor{Mahogany}{\small{\texttt{scope}}}, encapsulates the entire temporal requirement under the probabilistic operator. While this placement supports some types of probabilistic properties, it restricts the ability to express more nuanced formulations. Specifically, it prevents the specification of requirements in which the probabilistic constraint applies only to the consequent of a conditional statement (e.g., “If A, then with probability $\geq$ 0.9 eventually B”), i.e., to the \textcolor{blue}{\small{\texttt{timing}}}-\textcolor{violet}{\small{\texttt{response}}} portion of the requirement, rather than the entire temporal expression. Moreover, this placement introduces a risk of vacuous satisfaction, i.e.,  if the \textcolor{orange}{\small{\texttt{condition}}} rarely holds, and even if the \textcolor{violet}{\small{\texttt{response}}} always fails, the entire requirement could still be vacuously satisfied with high probability.

This led us to the decision to place the \textcolor{cyan}{\small{\texttt{probability}}} field directly before the \textcolor{blue}{\small{\texttt{timing}}} field. This placement is common in probabilistic logics including PCTL$^*$, and supports the specification of a wide range of meaningful probabilistic properties. It also offers flexibility --- when both \textcolor{Mahogany}{\small{\texttt{scope}}} and \textcolor{orange}{\small{\texttt{condition}}} are omitted, the probability applies to the entire requirement. Finally, when \textcolor{Mahogany}{\small{\texttt{scope}}} and \textcolor{orange}{\small{\texttt{condition}}} 
 are used, we avoid the vacuity issue described above since the probability governs only the \textcolor{blue}{\small{\texttt{timing}}}-\textcolor{violet}{\small{\texttt{response}}} portion of the requirement.

\section{Compositional Formalization}
\label{sec:compositional_formalization}

Our formalization algorithm (Algorithm~\ref{alg1}) is compositional, i.e., rather than manually defining a separate formula for each individual template key, we systematically construct complete PCTL$^*$  formulas by composing subformulas based on the values of the template fields. For example, consider requirement \textbf{[P-006]} and the construction of its corresponding template key. As outlined in the introduction, classic \fret template keys are derived from the values of the \textcolor{Mahogany}{\small{\texttt{scope}}}, \textcolor{orange}{\small{\texttt{condition}}}, and \textcolor{blue}{\small{\texttt{timing}}} fields. With our extensions, we introduce a fourth component, the \textcolor{cyan}{\small{\texttt{probability}}}  field, into the template key definition. Thus, the template key for \textbf{[P-006]} is: \emph{[in, holding, bound, next]}. To generate the corresponding PCTL$^*$  formula, the algorithm interprets and composes the semantics of the \textit{in} scope, \textit{holding} condition, probabilistic \textit{bound}, and \textit{next} timing.%  used in the requirement. 

Our algorithm builds on previous work by Giannakopoulou et al.~\cite{giannakopoulou2020generation,GIANNAKOPOULOU2021106590}, which introduced a compositional algorithm for generating \salt formulas from classic \fretish requirements, i.e., the non-highlighted parts of the grammar in Table~\ref{table:fret-full-grammar}, and translating them into LTL formulas. Our approach reuses the semantics of \fretish scope, timing and regular condition introduced in~\cite{GIANNAKOPOULOU2021106590}.\new{For consistency, we reuse the terminology and variable names from \cite{GIANNAKOPOULOU2021106590} where possible.}
We next briefly discuss PCTL$^*$, relevant \salt constructs, as well as \fretish scope intervals.

\subsection{Preliminaries}
\subsubsection{Probabilistic Computation Tree Logic Star}
\label{sec:preliminaries_pctl_salt}

PCTL$^*$~\cite{aziz1995usually,baieralgorithmic} is a probabilistic logic that is derived from branching tree logic CTL$^*$~\cite{emerson1986sometimes} by replacing the path quantifiers $\exists$ and $\forall$ with a probabilistic operator $P$. This
probabilistic operator defines an upper or a lower bound on the probability of the system evolution. As an example, the
formula $P_{\geq p}(\phi)$ is true at a given time, if the probability
that the future evolution of the system satisfies $\phi$ is at least
$p$~\cite{grunske09}.

Next, we briefly describe PCTL$^*$. For a detailed presentation of the logic, we refer the reader %to~\cite{hansson1994logic} for PCTL and 
to~\cite{aziz1995usually,baieralgorithmic}. %for PCTL*. 
For a finite set $AP$ of atomic propositions, a \textit{PCTL$^*$ state formula} $\Phi$ and a \textit{PCTL$^*$ path formula} $\Psi$ are defined by the following grammar:

\centerline{$ \Phi ::= true \mid a \mid \Phi_1 \wedge \Phi_2 \mid \neg \Phi \mid P_{\sim p }[\Psi]$}
\centerline{$\Psi ::= \Phi \mid X \Psi \mid \Psi_1 U \Psi_2 \mid \Psi_1 U^{\leq n} \Psi_2 \mid \Psi_1 \wedge \Psi_2 \mid \neg \Psi$ }

\noindent where $a \in AP$, $\sim\; \in \{<,\leq, >, \geq\}$, $p\in [0,1]$ and $n$ is a natural number. $X$, $U$, $U^{\leq n}$ are the temporal operators `next', `until' and `bounded until', respectively. The $X$ operator refers to the next time point, i.e., $X$ $\Psi$ holds iff $\Psi$ is true at the next time point. The time-bounded `until' formula $\Psi_1 U^{\leq n} \Psi_2$ requires that $\Psi_1$ holds continuously within time steps $[0, x)$, where $x \leq n$ and $\Psi_2$ becomes true at time point $x$. The untimed path formula $\Psi_1 U \Psi_2$ can be obtained by setting $n$ to $\infty$. Other Boolean operators, such as disjunction ($\lor$) and implication ($\Rightarrow$), can be expressed using the primary Boolean operators: negation ($\neg$) and conjunction ($\land$).
For path formulae, the temporal operators `eventually' $F$ and `globally' $G$ (as well as their time-bounded variants $F^{\leq k}$ and $G^{\leq n}$) can be derived from the `until' operator~\cite{grunske09}; $F\ \phi$ is true iff $\phi$ holds at the current or some future time point; $G\ \phi$ is true iff $\phi$ is always true in the future.
Note that  $\Psi$ can be any LTL expression, which allows arbitrary combinations of path formulas. %\footnote{In comparison with PCTL~\cite{hansson1994logic} where only path formulas of the form $X \Phi$,  $\Phi_1 U \Phi_2$, and $\Phi_1 U^{\leq k} \Phi_2$ are allowed. Thus, PCTL* is strictly more expressive than PCTL~\cite{baier2008principles}.}.

\subsubsection{Structured Assertion Language for Temporal Logic (SALT)}
We use \salt to generate PCTL$^*$  \emph{path formulas}, leveraging its propositional operators: {\small{\saltnot, \saltand, \saltor, \saltimplies}}, its temporal operators: {\small{\until, \always, \eventually}} and {\small{\saltnext}}, and timed modifier: {\small{\timed[$\sim$]}}, where $\sim$ is one of $<$ or $\leq$ for bounded temporal constraints (e.g., {\small{\once \timed[$\leq3$]}}). 
%
% Furthermore, we use \salt scope operators to restrict the scope of a requirement, which state that the requirement has to hold only {\small{\texttt{before}}} or {\small{\texttt{after}}} or {\small{\texttt{between}}} some events. When using scope operators it is mandatory to specify whether the delimiting events are part of the interval {\small{(\inclusive)}} or not {\small{(\exclusive)}}. Furthermore, for the scope operators it has to be stated whether the occurrence of the delimiting events is strictly {\small{\required}} (must eventually occur), {\small{\weak}} (may or may not occur) or {\small{\optional}} (the expression is only considered if the delimiting event eventually occurs). 
%
We also use \salt's scope operators to restrict when a requirement must hold, using {\small{\texttt{before}}}, {\small{\texttt{after}}}, or {\small{\texttt{between}}} events. Scope intervals must specify inclusion ({\small{\inclusive}}) or exclusion ({\small{\exclusive}}) of their delimiting events, as well as their occurrence mode: {\small{\required}} (must occur), {\small{\weak}} (may occur), or {\small{\optional}} (only evaluated if the event occurs).

\subsubsection{\fretish Scope Intervals}
\fret treats a scope as an interval between endpoints \leftend and \rightend, which depend on the scope option (see the functions called \emph{left(scope)} and \emph{right(scope)} in Table~\ref{table:allAlgTables}). The following abbreviations describe the endpoints: \FTP: first time point in execution, \FiM/\LiM: first/last state in mode, \FNiM/\LNiM: first/last state not in mode, \FFiM/\FLiM: first occurrence of \FiM/\LiM in execution.  
For example, \fretin scope is defined between points \FiM and \LiM. %The abbreviations are replaced by logic formulas characterizing the time points they refer to (Lines~\ref{line:forLoop1}-~\ref{line:forLoop2}).

We reuse these endpoints but observe that some of them coincide in the context of PCTL$^*$ (see the function called \emph{endpoint(acronym)} in Table~\ref{table:allAlgTables}). PCTL$^*$  path formulas are expressed in future time LTL and thus look ahead from left to right on a trace. As a consequence, they can only detect the left endpoint of an interval one time point before it occurs. For instance, \FiM (and \FFiM) is detected when \saltnot(\mode) \saltand \saltnext \mode holds, which coincides with \LNiM. Similarly, \LiM (and \FLiM) is detected when \mode \saltand \saltnot(\saltnext \saltnot \mode), which coincides with \FNiM. As such, most scope intervals are open on the left and closed on the right. Since we reason over infinite traces, the final time point of a trace is substituted by \textsc{false} (e.g., see \emph{null} and \emph{after} scopes in \emph{right(scope)}), making these intervals open on the right. Another exception is when \leftend is \FTP. Since \FTP is a point where the formulas are interpreted, the corresponding scope intervals are inclusive of \FTP and thus closed on the left.

\subsection{Compositional Formalization Algorithm}
\label{sec:algorithm}

Algorithm~\ref{alg1} describes the generation of PCTL$^*$ formalizations.  Our algorithm starts by retrieving, from Table~\ref{table:allAlgTables}, the \tFormula{} associated with the timing field (line~\ref{line:tformula}). For instance, for requirement \textbf{[P-006]}, in which we use \textit{next} timing, we retrieve the \salt formula: next incursionDetected, which means that incursionDetected must hold at the next timepoint. 
In the next step (line~\ref{line:formulaltl}), we enforce \tFormula{} to hold only within the interval delimited by the scope's~\rightend. 
For example, for \emph{in} scope, \ltlform{} returns: {\small{(next \response) \before \inclusive \weak \rightend}} which ensures that next \response must occur either before or exactly at the right endpoint of the scope interval, if that endpoint occurs. The \weak operator ensures that the constraint remains in effect even if the delimiting event \LiM never occurs. Note that \ltlform{} only concerns the scope’s \rightend; checking that the requirement is enforced within the scope’s \leftend is handled separately by $\mathit{sform}$ (see Table~\ref{table:allAlgTables}).

\begin{algorithm}[!ht]
\parbox{\linewidth}{ % Create a paragraph box for \small to work
\small
\begin{algorithmic}[1]
\Function{GenProbFormula}{$\mathit{scope},\mathit{cond},\mathit{prob},\mathit{timing}$}

\LeftComment{\emph{Timing formalization}}
   \State{\tFormula $\gets$ \tform $(\mathit{timing})$ \label{line:tformula}}
    \State{\ltlFormula$ \gets \SALT (\mathit{ltlform}(\mathit{scope}, \mathit{tFormula}))$ \label{line:formulaltl}}

\LeftComment {\emph{Probability formalization}}
\If{$\mathit{prob}$ = $\mathsf{null}$ \label{line:probform1}}
    \State{$\mathit{pFormula} \gets P_{\geq 1}[\mathit{ltlFormula}]$ \label{line:probform2}}
\ElsIf{$\mathit{prob}$ = bound \label{line:probform3}}
    \State{$\mathit{pFormula} \gets P_{\sim \bound}[\mathit{ltlFormula}]$\label{line:probform4}}
 \EndIf

\LeftComment {\emph{Condition formalization}}
\If{$\mathit{cond} = \mathsf{null}$} \label{line:conditionbeginning}
    \State{$\mathit{baseForm} \gets \saltap$}\label{line:baseline1} %\Comment{\saltap: atomic proposition}
\ElsIf{$\mathit{cond} = \mathsf{holding}$}
    \State{$\mathit{trigger} \gets \condition$}
    \State{$\mathit{baseForm} \gets\saltalways (\mathit{trigger}\; \saltimplies\; \saltap)$}\label{line:baseline2}
\ElsIf{$\mathit{cond} = \mathsf{regular}$}
    \State{$\mathit{trigger} \gets ((\saltnot\; \condition)\; \saltand\; \saltnext\; \condition)$} 
    \State{$\mathit{f1} \gets (\saltalways\; (\mathit{trigger}\; \saltimplies\;  \saltnext\; \saltap))$}
    \State{$\mathit{f2} \gets (\condition\; \saltimplies\; \saltap)$}    
    \State{$\mathit{baseForm} \gets \mathit{f1}\; \saltand\; \mathit{f2}$}\label{line:baseline3}
\EndIf

\LeftComment {\emph{Scope formalization}}
\State{$\mathit{formula} \gets \mathit{sform}(\mathit{scope}, \mathit{baseForm})$ \label{line:scopetable}}
\State{$\mathit{formula} \gets \textsc{Instantiate}(\mathit{formula}, \leftend, \mathit{left}(\mathit{scope}))$}\label{line:substL}
\State{$\mathit{formula} \gets \textsc{Instantiate}(\mathit{formula}, \rightend, \mathit{right}(\mathit{scope}))$}\label{line:substR}
\For{$\mathit{ac}\in\{\FiM, \FFiM, \LNiM, \LiM, \FNiM, \FLiM \}$}\label{line:forLoop1}
   \State{$\mathit{formula} \gets \textsc{Instantiate}(\mathit{formula}, \mathit{ac}, \mathit{endpoint}(\mathit{ac}))$}\label{line:forLoop2}
\EndFor
\State{$\mathit{formula} \gets \SALT(\mathit{formula})$}\label{line:salt2}
\State{$\mathit{formula} \gets \textsc{Instantiate}(\mathit{formula}, \saltap, \mathit{pFormula})$}\label{line:probformsubst}
\Return{$P_{\geq 1} [\mathit{formula}]$}\label{line:return}
\EndFunction
\end{algorithmic}
} %end parbox
\caption{Generation of Probabilistic Formulas}
\label{alg1}
\end{algorithm}

\begin{table}
\caption{Functions used in Algorithm~\ref{alg1} adapted from~\cite{GIANNAKOPOULOU2021106590}.}
\label{table:allAlgTables}

\vspace*{-3mm}
\begin{center}
\def\tabcolsep{3pt}
{\renewcommand{\arraystretch}{0.78}
\scalebox{0.8}{
\begin{tabular}{|p{2.2cm}p{7cm}|}
\hline
$\mathit{timing}$
 & $\mathit{tform}(\mathit{timing})$ \\
\hline \hline
{\em immediately} 
& \response\\
{\em next} 
&  next \response\\
{\em always} 
& \saltalways \response \\
{\em eventually} 
&  \salteventually \response \\
{\em until}  
&  (\response \saltuntil \exclusive \weak \stopcond)\\
{\em before}  
&  \response \saltreleases (\saltnot \stopcond)\\
{\em for \durn} 
&  (\saltalways \timed[$\leq$\durn] \response)\\
{\em within \durn}
&  (\salteventually \timed[$\leq$\durn] \response)\\
{\em never}
& {\em always}(\saltnot \response)\\
{\em after \durn}
&  {\em for}(\dur, \saltnot \response) \saltand\ {\em within}(\durn$+1$, \response) \\
\hline
\end{tabular}}

\vspace*{1.5mm}
\scalebox{0.8}{
\begin{tabular}{|l l|}
\hline
 $\mathit{scope}$ %& \textbf{args} 
 & $\mathit{ltlform}(\mathit{scope}, \mathit{tFormula})$ \\
\hline \hline
{\em null, after} %&  (\formula) 
& $\mathit{tFormula}$\\
{\em in, before, notIn} %&  (\formula) 
& ($\mathit{tFormula}$) \before \inclusive \weak \rightend \\
{\em onlyIn} %&  (\formula) 
& \saltnot ($\mathit{tFormula}$) \before \inclusive \weak \rightend\\
{\em onlyBefore} %&  (\formula) 
& \saltnot ($\mathit{tFormula}$)\\
{\em onlyAfter} %&  (\formula) 
& \saltnot ($\mathit{tFormula}$ \before \inclusive  \weak \rightend)\\
\hline
\end{tabular}}

\vspace*{1.5mm}
\scalebox{0.8}{
\begin{tabular}{|l l|}
\hline
$\mathit{scope}$
 & $\mathit{sform}(\mathit{scope}, \mathit{baseForm})$ \\
\hline \hline
{\em null}
& $\mathit{baseForm}$\\
{\em in}
& (\saltalways\\ &
\hspace{\parindent} \hspace{\parindent} \hspace{\parindent} \hspace{\parindent}($\mathit{baseForm}$ \between   
 \exclusive \optional \leftendnospace,  \\ 
& \hspace{\parindent} \hspace{\parindent} \hspace{\parindent} \hspace{\parindent}  \inclusive \weak \rightendnospace))\\
& \saltand \\
&  \hspace{\parindent} \hspace{\parindent} \hspace{\parindent} \hspace{\parindent}	(\mode \saltimplies ($\mathit{baseForm}$ \\ & \hspace{\parindent} \hspace{\parindent} \hspace{\parindent} \hspace{\parindent}	 \before \inclusive \weak \rightendnospace))\\ 
{\em before} 
& \mode \saltor ($\mathit{baseForm}$ \before \inclusive \weak \rightendnospace) \\
{\em after} 
& ($\mathit{baseForm}$ \after  \exclusive \optional \leftend) \\
{\em notIn} 
& {\em{in}}($\mathit{baseForm}$)\\
{\em onlyIn}  
& {\em{in}}(\saltnot($\mathit{baseForm}$))\\
{\em onlyBefore} 
& ((\saltnot \mode) \saltimplies ({\em{after}}(\saltnot($\mathit{baseForm}$)))) \\
 & \hspace{\parindent} \hspace{\parindent} \hspace{\parindent} \hspace{\parindent} \saltand  (\mode \saltimplies (\saltnot($\mathit{baseForm}$))) \\
 {\em onlyAfter} 
& (\saltnot($\mathit{baseForm}$) \before \inclusive \weak \rightendnospace) \\
\hline
\end{tabular}}

\vspace*{1.5mm}
\scalebox{0.8}{
\begin{tabular}{|p{3cm}p{3cm}p{3cm}|}
\hline
$\mathit{scope}$ & $\mathit{left}(\mathit{scope})$ & $\mathit{right}(\mathit{scope})$ \\
\hline \hline
{\em null} & \FTP & \textsc{false} \\
{\em before} & \FTP & \FFiM \\
{\em after} & \FLiM & \textsc{false} \\
{\em in} & \FiM & \LiM \\
{\em notIn, onlyIn} & \FNiM & \LNiM \\
{\em onlyBefore} & \FFiM & \textsc{false} \\
{\em onlyAfter} & \FTP & \FLiM \\
\hline
\end{tabular}}

\vspace*{1.5mm}
\scalebox{0.8}{
\begin{tabular}{|p{4cm}p{6cm}|}
\hline
 $\mathit{acronym}$ %& \textbf{args}
 & $\mathit{endpoint}(\mathit{acronym})$ \\
\hline \hline
{\em \FiM, \FFiM, \LNiM } %&  (\mode) 
& (\saltnot \mode) \saltand (\saltnext \mode)\\
{\em \LiM, \FNiM, \FLiM } %&  (\mode) 
& (\mode) \saltand (\saltnext \saltnot \mode)\\
\hline
\end{tabular}}
}
\end{center}
\vspace{-3mm}
\end{table}

Next, we invoke \salt, providing the \tFormula{} along with the scope option as input, and receive an LTL formula as output. For \textbf{[P-006]}, \salt returns {\small{\texttt{(X incursionDetected) \& (!RIGHT)}}}, meaning that incursionDetected must hold at the next time point, unless the current time point is the last one (\rightend) within the scope interval.

On lines \ref{line:probform1}--\ref{line:probform4}, \pFormula{} is formed by taking the probability field into consideration. When the probability field is null, \pFormula{} equals \ltlFormula{} inside the PCTL$^*$ \emph{P} operator with a probability greater than or equal to $1$. When the probability field is not null and thus there is a specified probabilistic \bound, \pFormula{} is captured by applying the probabilistic \emph{P} with the predefined \bound operator on \ltlFormula{}. For requirement \textbf{[P-006]}, line~\ref{line:probform4} returns {\small{$P_{\geq 0.99}$[(X RES) \& (!RIGHT)]}}.

Next, \baseForm{} is constructed by considering the condition field (lines~\ref{line:conditionbeginning}--\ref{line:baseline3}). If the condition field is null, \baseForm{} is set to \saltap, a placeholder that is later instantiated to \pFormula{} (line~\ref{line:probformsubst}). 
Conditions can be interpreted in two ways 
depending on whether the condition field is set to \emph{holding} or \emph{regular}.
If the condition value is \emph{holding}, \baseForm{} enforces the formula at any point where a \trigger{} is detected, where \trigger{} corresponds to the condition (\condition) \emph{being} true. For {\textbf{[P-006]}, line~\ref{line:baseline2} returns {\small{\always (incursion \saltimplies\ \saltap)}}.
If the condition value is \emph{regular}, \baseForm{} is defined as a conjunction: the first conjunct imposes the formula whenever a \trigger{} is detected (when a condition \emph{becomes} true, i.e., transitions from false to true, as described by {\small{(\saltnot \condition) \saltand \saltnext \condition)}}, and the second conjunct covers the case where the \trigger{} occurs at the beginning of the interval.

Continuing, the \sform{} function uses \salt scope operators to enforce \baseForm{} over all intervals that are relevant to the specified scope field (line~\ref{line:scopetable}). Note that when \leftend is \FTP, there is no need to constrain the application of \leftend within \sform{} (e.g., scope \fretbefore). For scope \fretin, \baseForm{} is expressed as a conjunction. The first conjunct states that at any point in an execution trace (\saltalways), \baseForm{} must hold between \leftend and \rightend, thus imposing \baseForm{} on all relevant intervals in the execution. The second conjunct applies to cases where \modevariable holds at the beginning of the trace, i.e., the first relevant interval starts at \FTP; it states that in this case, \baseForm{} must hold until \rightend. This special case is required because the interval starting at \FTP is not open on the left, as explained earlier. For \textbf{[P-006]}, line~\ref{line:scopetable} returns: {\small{(\always (\always (incursion \saltimplies\ \saltap) \between \exclusive \optional \leftend, \inclusive \weak \rightend)) \saltand (auto\_takeoff \saltimplies (\saltnot (\always (incursion \saltimplies\ \saltap)) \before \inclusive \weak \rightend))}}.

Next, endpoints \leftend and \rightend are substituted (lines~\ref{line:substL}–\ref{line:forLoop2}) with the appropriate expressions for the given scope, as specified in Table~\ref{table:allAlgTables}. For \textbf{[P-006]}, \leftend} (\FiM) is replaced with 
{\small{(\saltnot auto\_takeoff \saltand (\saltnext auto\_takeoff))}}, and \rightend (\LiM) is replaced with 
{(\small{auto\_takeoff \saltand (\saltnext \saltnot auto\_takeoff))}}.
On line~\ref{line:salt2}, \salt is invoked to generate the LTL expression. Note that the placeholder \saltap is not replaced by the probabilistic \pFormula{} until line~\ref{line:probformsubst}. The algorithm wraps the resulting formula in the PCTL$^*$  \texttt{P} operator and returns it on line~\ref{line:return} (see \textbf{[P-006]} in Table~\ref{table:autonomy_examples_filtered}).
% , e.g., for
% {\small{\textbf{[P-006]}:}} 
% \setlength{\fboxrule}{0.4pt} % Adjust for desired thinness (e.g., 0.4pt, 0.2pt)
%     \setlength{\fboxsep}{1pt}   % Adjust for desired padding around the image
%     \fcolorbox{gray}{white}{ % Frame color is gray, background is white
% {\scriptsize \texttt{P>=1[((G ((! ((! auto\_takeoff) \& (X auto\_takeoff))) | \hfill \break
% (X ((auto\_takeoff \& (X (! auto\_takeoff))) R (incursion => (P>0.99[((auto\_takeoff \& (X (! auto\_takeoff))) | ((X incursionDetected) \& (! (auto\_takeoff \& (X (! auto\_takeoff))))))])))))) \& (auto\_takeoff => ((auto\_takeoff \& (X (! auto\_takeoff))) | ((auto\_takeoff \& (X (! auto\_takeoff))) R (incursion => (P>0.99[((auto\_takeoff \& (X (! auto\_takeoff))) | ((X incursionDetected) \& (! (auto\_takeoff \& (X (! auto\_takeoff))))))]))))))]}}
% }

\subsection{Algorithm Correctness: Termination and Well-Formedness}

%\gricelnote{use same newcommands as in algorithm in proof. Fix descritpion of proof based on new algorithm. Remove old newcommands.}

%We prove that Algorithm~\ref{alg1} always terminates and  returns a well-formed PCTL* formula. 

\begin{thm}
    Procedure \textsc{Create\_Probabilistic\_Formula} (Algorithm~\ref{alg1}) terminates and returns a well-formed PCTL* formula (syntax defined in Section~\ref{sec:preliminaries_pctl_salt}) for any valid combination of input fields scope, condition, probability and timing.\end{thm}

%\vspace*{-4mm}
\begin{proof}
We prove the correctness of \textsc{Create\_Probabilisti{\-}c\_Formula} by direct proof of the steps in Algorithm~\ref{alg1}. \smallskip

\noindent\emph{\textbf{Timing (lines 2--3):}} We begin by showing that \ltlFormula{} is a valid LTL formula. For this, it is sufficient to show that the input provided to the \SALT{} function (line 3) is a valid \SALT{} formula. This function invokes the \SALT{} tool~\cite{SALT} which produces syntactically correct LTL formulas.  
%which translates this into a syntactically correct LTL one. 
We systematically construct the input SALT formula via the composition of \tform{} and~\ltlform{}. Since both functions (Table~\ref{table:allAlgTables}) are finite, we  exhaustively enumerate all possible combinations. This allows us to verify, by construction, that the output of each composition yields a syntactically valid \SALT{} formula. On line~\ref{line:tformula}, all possible values of \tform{} %(a total of ten) 
are valid \SALT{} formulas. In Table~\ref{table:allAlgTables} (and the generated \SALT{} formulas), \rightend, \response and \stopcond are treated as atomic propositions that are later instantiated. 
We next consider the \ltlform{} output (line~\ref{line:formulaltl}) from Table~\ref{table:allAlgTables}. The first row outputs the valid \SALT{} formula, \tFormula{}. This is valid as it is the output from line 2 (discussed above). The subsequent rows extend this using other \SALT{} constructs such as \texttt{\small{before inclusive weak RIGHT}} in the second row. Each of these produce a valid \SALT{} formula. \smallskip  %The negation of a formula is also a valid \SALT{} formula, hence rows four and five are valid. Finally, row three follows the same rules (negation and extending with \texttt{\small{before inclusive weak RIGHT}}) resulting in a valid \SALT{} formula. 

\noindent\emph{\textbf{Probability (lines 4--7):}} Based on the probability input, either line~\ref{line:probform2} or \ref{line:probform4} produces a syntactically valid PCTL$^*$  formula (\probForm), of the form $P_{\sim p}[\varphi]$, where $\varphi$ is a well-formed LTL formula.\footnote{We remind the reader of the LTL syntax~\cite{pnueli1977temporal}: \( \varphi ::= true \mid a \mid \varphi_1 \wedge \varphi_2 \mid \neg \varphi_1 \mid X \varphi_1 \mid \varphi_1 U \varphi_2 \mid \varphi_1 U^{\leq n} \varphi_2 \), where $\varphi_{1}$ and $\varphi_{2}$ are LTL formulas.} 
The correctness of this construction follows from the observation that the probabilistic operator $P$ in PCTL$^*$  is defined over a path formula $\Psi$, and that every LTL formula $\varphi$ is syntactically a path formula---i.e., the LTL syntax is a strict syntactic subset of the path formula fragment $\Psi$. \smallskip

\noindent\emph{\textbf{Condition (lines 8--17):}} Depending on the condition input, \baseForm{} is assigned one of three valid \SALT{} formulas from lines \ref{line:baseline1}, \ref{line:baseline2} or \ref{line:baseline3}. We treat \texttt{COND} and \texttt{PROBFORM} as atomic propositions that are instantiated later. \smallskip

\noindent\textbf{\emph{Scope (lines 18--25):}} %lines \ref{line:scopetable}--\ref{line:forLoop2} generate a correct \SALT{} formula. 
The function, \sform{}, (line 18) first returns a fragment of a \SALT{} formula. Then, \INSTANTIATE{} replaces the keywords \LEFT{} (line 19), \RIGHT{} (line 20) and any abbreviations in \formula{} for \SALT{} fragments (lines 21--22). 
This results in a correct by construction \SALT{} formula and can be proved similar to \textit{Timing} on lines \ref{line:tformula}--\ref{line:formulaltl}. For brevity, we refer the reader to the compositional algorithm for classic FRET~\cite{GIANNAKOPOULOU2021106590} for detailed explanations of the acronyms in Table \ref{table:allAlgTables}. On line~\ref{line:salt2}, \formula{} is replaced by an LTL formula after calling \SALT. Notice that this LTL formula contains \saltap (from lines 9, 12 or 15) treated as an atomic proposition. Hence, \formula{} on line \ref{line:probformsubst} complies with the syntax: {\small{\( \varphi' ::= true \mid a \mid P_{\sim \bound}[\text{\formulaltlnospace}] \mid \varphi_1 \wedge \varphi_2 \mid \neg \varphi_1 \mid X \varphi_1 \mid \varphi_1 U \varphi_2 \mid \varphi_1 U^{\leq n} \varphi_2 \)}}. %%Since probabilistic model checking is limited to the verification of state formulas, it is not possible to directly verify properties expressed in the last three forms (i.e., $X$, $U$ and $U^{\leq n}$). 
To ensure that any generated formula conforms to the syntax of a valid PCTL$^*$  state formula, line \ref{line:return} encapsulates it within the probabilistic operator \textit{P}, resulting in a \formula{} that adheres to the PCTL$^*$  strict syntactic subset:
\centerline{$ \Phi_{\text{\formula{}}} ::= P_{\geq 1 }[\varphi']$}\smallskip
 
\noindent\textbf{\emph{Termination:}} All operations are deterministic and rely on finite derivations from the tables, conditional branches, and variable replacements. There is no recursion and our only loop on line 21 has exactly 6 iterations for the set of acronyms. Hence, every control path leads to a return statement, and assuming that the \SALT{} translation function terminates (see \SALT{} semantics in~\cite{SALT}), Algorithm~\ref{alg1} is guaranteed to terminate.\smallskip

By analyzing the formula construction steps, along with a correctness-by-construction argument for the intermediate formulas derived from Table \ref{table:allAlgTables} and the \texttt{SALT} translation function, we have shown that the final output from Algorithm~\ref{alg1} is guaranteed to be a well-formed PCTL$^*$ formula.

\end{proof}
\section{Validation of Formalization Semantics}
\label{sec:testing}

\new{As previously stated, we adopted PCTL$^*$ because it provides a formal syntax and semantics for the specification of probabilistic requirements, such as ``the probability of the system detecting a violation correctly (from a given initial system state) is at least 96\%'' ($P_{\ge 0.96}[F\ violation_\_detected]$). The automated verification of such properties is supported by probabilistic verification tools such as PRISM~\cite{PrismPropSpec} and Storm~\cite{dehnert2017storm}.}
%\marienote{add prism and storm references above}

To enhance confidence that the generated PCTL$^*$ formulas capture the intended semantics, we employ a dedicated validation framework. For each template key and its corresponding PCTL$^*$  formula, the framework checks that the formula conforms to the template key's semantics.  As part of the validation process, we generate Discrete-Time Markov Chains (DTMCs) to serve as test models against the generated formulas. 

\new{Our validation approach is inspired by established software-engineering techniques, such as oracle-based and coverage-driven testing. To ensure comprehensive coverage, we generate test cases for all $560$ \fret template-keys, exercising multiple field valuations for each. Our approach draws upon validation approaches used for Prospec~\cite{Prospec} and classic FRET \cite{GIANNAKOPOULOU2021106590} (both for LTL), which rely on trace generation and Spin/NuSMV model checking. However, our approach introduces a critical extension: we construct DTMC models and employ PRISM for probabilistic model checking, thereby bridging established validation techniques with the probabilistic domain. While non-exhaustive, our method systematically explores the template-key space and allows us to reuse the classic FRET oracles to ensure semantic consistency with the classic FRET's LTL-based definitions.}

We begin with a brief overview of Discrete-Time Markov Chains (DTMCs) before detailing our validation approach.

\subsection{Preliminaries: Discrete-Time Markov Chains (DTMCs)}
A Discrete-Time Markov Chain (DTMC) describes a system that transitions between a finite countable set of states in discrete time steps~\cite{kwiatkowska2007stochastic}. 
A DTMC is formally defined as a tuple $\langle S, \overline{s}, P, L \rangle$, where $S$ is a finite set of states, $\overline{s} \in S$ is the initial state, $P: S \times S \rightarrow [0,1]$ is a transition probability matrix, and $L: S \rightarrow 2^{AP}$ is a function that labels states with atomic propositions in $AP$. Note that $\Sigma_{s'\in S} P(s,s') = 1$ for all $s \in S$. If the DTMC is currently in state $s$, then
$P(s, s')$ defines the probability of transitioning to $s'$ within one time step. A path through a DTMC is a finite or infinite sequence of states that
describes a possible execution of the modeled system.

\subsection{Automated Validation Framework}
We describe the three main components of the framework.
The \textsf{\small{Probabilistic Model Generator}} (Figure \ref{fig:architecture}) component uses a random approach for the generation of DTMC models; an example is shown in Figure~\ref{fig:dtmcTest1}.
Each of the states $s0$--$s9$ of the model is defined by the state variable tuple ($t$, $m$, $c$, $sc$, $r$) representing time, \mode, \condition, \stopcond and \response, respectively. Time~$t\in [0,\hbox{\textsl{Max}}]$, $\hbox{\textsl{Max}}\in\mathbb{N}$ counts the number of steps needed to reach a state s$\in$\{s0, s1, \ldots, s9\}. For example, 5 steps are needed to reach state s6, hence $t=5$ in s6. 
The rest of the state variables are Booleans $m,c,sc,r\in [0,1]$.
E.g., in state s9, \condition and \stopcond are 0, while \mode 
and \response are 1.

In order to generate DTMCs, our tool first randomly chooses a number for $\hbox{\textsl{Max}}$. For illustration, in our example $\hbox{\textsl{Max}}=6$. Second, it randomly selects between 0 and 3  (random, disjoint, non-consecutive) intervals ranging from 0 to $\hbox{\textsl{Max}}$, for each \mode, \condition, \stopcond and \response.   In our example, two intervals were constructed for \mode, i.e., intervals [0,1] (states $s0-s1$) and [3,6] (states $s3-s9$), a single interval for \condition, i.e., [3,5]) (states  $s3-s7$), and zero intervals were constructed for \stopcond. 
The construction of \response intervals takes a different approach; we select random intervals as before, but since responses may be probabilistic, we assign $r$ to be 0 and 1 in separate branching states. In the example, the interval [4,6] was chosen, from which probabilistic transitions, from s3 to $s4-s5$, from s4 to $s6-s7$, and so on, were generated.
Finally, our tool generates a random transition probability matrix. %It only considers one decimal point for each transition probability. 
%This approach produces a very large number of distinct traces.

\begin{figure} [t]
\centering
\includegraphics[width=0.8\textwidth]{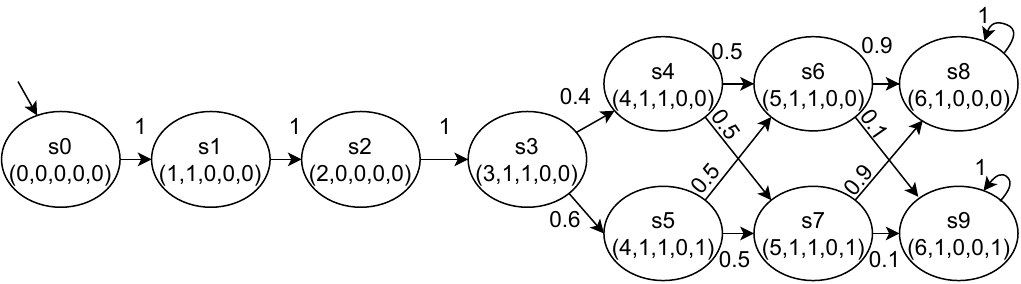}
\caption{\label{fig:dtmcTest1} \new{An example of a generated DTMC used for validation. States are defined by a 5-tuple ($t,m,c,sc,r$) representing time, \mode, \condition, \stopcond and \response.  The first element, $t$, tracks time by counting transitions from the initial state. The other four elements are boolean indicators for \mode ($m$), \condition ($c$), \stopcond ($sc$), and \response ($r$), where 1 signifies that the proposition is true in that state and 0 signifies it is false.}}%The initial state is indicated with a dashed border. }
\end{figure}

The \textsf{\small{Probabilistic Oracle}} component interprets the semantics of a template key on a DTMC and produces an expected value which can be either \texttt{true} or \texttt{false}. Our DTMCs contain no loops (except for the self-loops in the final states) and thus, each model represents a finite number of paths. 

Our approach is as follows. For each DTMC:
\begin{enumerate}
    \item Enumerate all possible paths.
    \item Apply existing\footnote{We have also created a new non-probabilistic oracle for holding conditions.} classic FRET oracles~\cite{GIANNAKOPOULOU2021106590} for the template key without the probability field. 
    If the oracle returns \texttt{true}, retain the path; otherwise discard it. 
    %\anastasianote{Stopped here.}
    \item For the resulting subset, calculate the sub-path probabilities of each retained path. The probability of a sub-path is the product of the transition probabilities of the sub-path. %Notice that for formulas without a mode and condition, we calculate probabilities of sub-paths that start from the initial state, while for other formulas with conditions, we calculate probabilities of sub-paths from states where the mode holds and/or the condition holds (holding condition) or becomes true from false (regular condition).
\end{enumerate}

The probabilistic oracle is shown in Algorithm~\ref{alg:prob-oracle-code}, which calculates the probability of a sub-path through the retained paths in step 3 above in a compositional manner. 
It computes sub-paths that represent the sequence of states starting at a state where the mode and the condition (if any) hold, up to the conclusion of the response. The probability of each sub-path must satisfy the requirement's probability constraint. 

\begin{algorithm}[t]
\parbox{\linewidth}{ % Create a paragraph box for \small to work
\small
\begin{algorithmic}[1]
% \scriptsize
\Function{CheckProbability}{$\mathit{dtmc}, \mathit{key}, \mathit{probBound}$}
    \State $\mathit{badPaths} \gets \emptyset$; $\mathit{goodPaths} \gets \emptyset$
    
    \ForAll{$\mathit{path} \in \mathit{dtmc.paths.filter(\lambda p.}$$\mathsf{FREToracle}(\mathit{p}, \mathit{key})$}
        \ForAll{$l \in \mathsf{starts}(\mathit{path}, \mathit{key.scope}, \mathit{key.condition})$}
            \State $e \gets \langle \text{Table3} \rangle (\mathit{path}, l, \mathit{key.timing})$
            \State $e \gets \mathsf{getRightIndex}(\mathit{path}, l, \mathit{key}.\mathit{timing})$ \Comment{See Table~\ref{tab:rightindex}}
            \State $p \gets \mathsf{reduce}(*, [l+1 \ldots e].\mathsf{map}(\mathit{dtmc.prob}))$
            \If{$\mathit{probBound}(p)$}
                \State $\mathit{goodPaths} \gets \mathit{goodPaths} \cup \{\mathit{path}\}$
            \Else
                \State $\mathit{badPaths} \gets \mathit{badPaths} \cup \{\mathit{path}\}$
            \EndIf
        \EndFor
    \EndFor
    \Return $[\mathit{goodPaths}, \mathit{badPaths}]$
\EndFunction
\end{algorithmic}
} %end box
\caption{Probabilistic Oracle}
\label{alg:prob-oracle-code}
\end{algorithm}

For a retained path through the DTMC, let $s$ be its sequence of states $s_i$  for $0 \leq i < length(s)$. Let $last = length(s) - 1$ and let $prob_i$ be the probability of taking the transition from state $s_{i-1}$ to state $s_i$, where $prob_0 = 1$. The probability of a sub-path is thus $\prod_{i = l+1}^{e} prob_i$ (line~6), where $l$ is the left \new{(starting)} index and $e$ is its right \new{(ending)} index of the sub-path . %Line~6 of the listing computes this product.

The starting index, $l$, is determined by the type of condition and mode.
When there is no condition or mode, $l$ is 0.
When, for example, there is a holding condition $c$ and an in mode $m$, there may be several sub-paths, each starting at an index $i$ where $s_i\models c \wedge m$. In the case of a regular condition $c$ and an in mode $m$, there may be several sub-paths, each starting at index $i$ where $s_i\models c \wedge m$ and either $i = 0$ or  $s_{i-1}\models \neg c$. Sub-paths for different combinations of modes and conditions are similarly defined. %Due to space limitations, we present these in~\cite{TODO}.\anastasianote{add} 
\new{The function \textsf{starts} (line~4)
in the listing computes the set of left indices for a given path.}

% The ending index of a sub-path, $e$, is determined by the starting index and the type of timing. We write $r$ for \response, $t$ for the time bound in the \textit{within} and \textit{for} timings, and $sc$ for \stopcond in the \textit{before} and \textit{until} timings.
% %Table~\ref{tab:rightindex} shows for each timing the constraint that determines the right index of a sub-path.
% \new{Table~\ref{tab:rightindex} provides the formulas for computing the right index based on each timing type. The function \texttt{getRightIndex} (line~6) looks up the appropriate formula from this table and computes $e$ accordingly.}

\new{The ending index $e$ is determined by the starting index $l$ and the timing constraint type. We write $r$ for the response, $t$ for the time bound in the \textit{within} and \textit{for} timings, and $sc$ for the stop condition in the \textit{before} and \textit{until} timings.
The function \textsf{getRightIndex} (line~6) takes as input the path, the starting index $l$, and the timing constraint from $\mathit{key.timing}$. It uses the timing constraint to look up the corresponding row in Table~\ref{tab:rightindex}, then applies the formula from that row to compute the right index $e$. For example, if the timing is ``next $r$'', the function looks up the second row of the table and computes $e = \min(last, l+1)$. For more complex timings like ``eventually $r$'', the formula determines $e$ as the first index where $r$ holds, ensuring no earlier state in the sub-path satisfies $r$.}

% \begin{table}[tb]
% \begin{center}
% \caption{Right sub-path index for timings constraints. }
% \label{tab:rightindex}
% \scalebox{0.9}{
% \begin{tabular}{ll}\toprule
% Timing & Right index ($e$)\\ \midrule
% immediately $r$ & $e = l$\\ %\hline
% next $r$ & $e = min(last, l+1)$\\ %\hline
% always $r$ & $e = last$\\ %\hline
% never $r$ & $e = last$\\ %\hline
% eventually $r$ & $l\leq e \land (s_e \models r) \land \forall k\, .\, l\leq k < e \Rightarrow s_k \models \neg r$\\ %\hline
% until $sc$ $r$ & $l\leq e\land 
% (((\forall k\,.\, l\leq k \leq last\Rightarrow s_k\models\neg sc)
% \land e = last)$\\
%  &$\lor ( (s_{e+1}\models sc)\land \forall k\,.\, l\leq k \leq e\Rightarrow s_k\models\neg sc))$\\ %\hline
%  before $sc$ $r$ & $l\leq e \land (s_e\models r) \land \forall k\, .\, l\leq k < e \Rightarrow s_k \models \neg r$\\ %\hline
% for $t$ $r$& $e = min(last, l+t)$\\ %\hline
% within $t$ $r$ & $l\leq e\leq l+t\land (s_e\models r) \land \forall k\, .\, l\leq k < e\Rightarrow s_k\models\neg r$\\ %\hline
% after $t$ $r$ & $e = min(last, l+t+1)$\\ 
% \bottomrule
% \end{tabular}}
% \end{center}
% \end{table}

\begin{table}[tb]
\begin{center}
\caption{\new{Computation of right sub-path index. The function $\mathsf{getRightIndex}(\mathit{path}, l, \mathit{timing})$ looks up the appropriate formula based on the timing type and computes the right index, $e$. The first column shows the different timing options for \fret requirements. The second column provides the formula to compute the right index $e$ of a sub-path starting at left index $l$ in $\mathit{path}$. Here, $last$ denotes the final index of $\mathit{path}$, $s_i$ denotes the state at index $i$ in $\mathit{path}$, $r$ is the response proposition, $sc$ is the stop condition, and $t$ is a time bound.}}
\label{tab:rightindex}
\scalebox{0.9}{
\begin{tabular}{ll}\toprule
\new{timing} & right index (denoted by $e$)\\ \midrule
immediately $r$ & $e = l$\\ 
next $r$ & $e = \min(last, l+1)$\\ 
always $r$ & $e = last$\\ 
never $r$ & $e = last$\\ 
eventually $r$ & $l\leq e \land (s_e \models r) \land \forall k\, .\, l\leq k < e \Rightarrow s_k \models \neg r$\\ 
until $sc$ $r$ & $l\leq e\land 
(((\forall k\,.\, l\leq k \leq last\Rightarrow s_k\models\neg sc)
\land e = last)$\\
 &$\lor ( (s_{e+1}\models sc)\land \forall k\,.\, l\leq k \leq e\Rightarrow s_k\models\neg sc))$\\ 
 before $sc$ $r$ & $l\leq e \land (s_e\models r) \land \forall k\, .\, l\leq k < e \Rightarrow s_k \models \neg r$\\ 
for $t$ $r$& $e = \min(last, l+t)$\\ 
within $t$ $r$ & $l\leq e\leq l+t\land (s_e\models r) \land \forall k\, .\, l\leq k < e\Rightarrow s_k\models\neg r$\\ 
after $t$ $r$ & $e = \min(last, l+t+1)$\\ 
\bottomrule
\end{tabular}}
\end{center}
\end{table}

For the  \emph{until} constraint  %Recall that in this case we are determining sub-paths of a path that was accepted by FRET's existing \emph{until} oracle.
the first disjunct deals with the case that the stop condition $sc$ never holds; hence the response $r$ must hold from index $l$ to the end of the path $last$. %; i.e., a weak until. 
The second disjunct defines $e$ as the last time the response $r$ must hold, which is right before the stop condition $sc$ holds. 

The \textsf{\small{Probabilistic Semantics Evaluator}} 
 uses the \prism model checker. It receives as input a DTMC model, a PCTL* formula, as well as the expected value provided by the \textsf{\small{Probabilistic Oracle}}.  It encodes the DTMC model in the \prism language and evaluates the truth value of the formula. Finally, it checks if the truth value that is returned by \prism agrees with the expected value produced by the \textsf{\small{Probabilistic Oracle}}.

% \subsection{Validation highlight}
 \medskip\noindent\textbf{Validation highlight.}
Despite our experience with formal logics, our validation framework proved essential for catching errors in the formulas that were generated by our algorithm. This highlights the fact that authoring formulas that are intended for direct use by analysis tools should not solely rely on  manual effort. 

One subtle error that it uncovered involved the initial version of the \salt formulas that was returned by the \ltlform{} function, where we incorrectly used the \optional modifier in place of the intended \weak modifier.  E.g., \ltlform{}  originally returned  {\small{($\mathit{tFormula}$) \before \inclusive \optional \rightend}} for \fretin scope.

At first glance, this formulation may appear plausible, as both \optional and \weak suggest that the right endpoint of the scope interval might not be reached. However, their semantics differ in subtle yet crucial ways. The \optional modifier introduces a conditional constraint: the specified behavior must hold if the endpoint occurs, but no constraint is imposed if it does not. In contrast, the \weak modifier requires the property to hold regardless of whether the endpoint is eventually reached. For the \fretin scope, which permits the interval to never end, \optional imposes an insufficient constraint, resulting in an under-approximation of the intended behavior.

This discrepancy could be easily missed during manual inspection, as the formulas are structurally similar. Our validation framework allowed us to identify and correct the issue.

\section{Evaluation}
\label{sec:evaluation}

%\anastasianote{Create supplementary material--some of the industrial case study formalisations can be moved to appear only there}

The extensions that we have presented in this work were motivated by requirements from publicly-available industrial case studies including reliability requirements in firefighting scenarios~\cite{pressburger2023wildfire} and autonomous system requirements~\cite{benz2024troupe,asaadi2020assured}. Next, we evaluate how \textit{effective} \fretish is in the specification of requirements \emph{in-the-wild} by investigating two research questions (inspired by~\cite{menghi2019specification}). We also assess the \textit{practical utility} of \fretish template keys in a third research question.
\begin{itemize}
    \item[\textbf{RQ1}:] To what degree can probabilistic requirements from the literature be expressed in extended \fretish?
    \item[\textbf{RQ2}:] Can extended \fret be used to specify requirements in a complex, industrial scenario?% defined in collaboration with our industrial partner?
    \item[\textbf{RQ3}:] Which template keys are most commonly used in the requirements gathered for \textbf{RQ1} and \textbf{RQ2}?
\end{itemize}

\noindent Note that we do not focus on the execution time of our tool, as the translation is inherently \textit{fast}, typically completely within milliseconds, due to its cache-retrieving-based mechanism. \medskip

%We present the conducted experiments and results for each research question.
%\noindent \textit{Experimental Setup.} 
\noindent\textbf{Experimental Setup.}
To answer each of the above research questions, we curated four sets of requirements as follows: \begin{compactitem}
    \item[\textsc{Set1}:] This requirement set was obtained via a detailed survey of multiple industrial real-world requirements from diverse transport sectors and defence \cite{grunske2008specification}. We acquired this set directly from the authors of \cite{grunske2008specification} (133 requirements).%\footnote{Note that they could not share confidential industry requirements with us.}. 
    \item[\textsc{Set2}:] We conducted an extensive literature review of the 239 peer-reviewed papers that cited \cite{grunske2008specification}, according to Google Scholar\footnote{\url{https://scholar.google.com/}}, extracting 199 natural language and formalized probabilistic requirements.
    \item[\textsc{Set3}:] We examined the 38 PCTL* properties that are publicly available in the Quantitative Verification Benchmark Set (QVBS)~\cite{hartmanns2019quantitative} which  provides an extensive repository of formally specified requirements, associated probabilistic models, and verification results across various tools.  
    \item[\textsc{Set4}:] We elicited a suite of 21 requirements in collaboration with our industry partner, RTRC. These requirements were used for verification against a probabilistic model of their autonomous aircraft taxi, take-off and landing system. 
\end{compactitem}
 These sets comprise both natural-language and formalized requirements from the literature, public repositories and industrial use cases. They span multiple domains, including biological processes and medical tasks, software applications and cyber physical systems. We focus on requirements that are expressible in PCTL$^*$, the target logic of our work. Consequently, we exclude properties that rely on constructs from other logics, such as CSL's steady-state operator or RPCTL's rewards. %This diversity provides a broad basis for evaluating our research questions.

\subsection{Expressing Probabilistic Requirements In-The-Wild (RQ1)}

\textbf{RQ1} investigates the extent to which our \fretish extension captures existing publicly available probabilistic requirements. To answer this, we examine \textsc{Set1}, \textsc{Set2} and \textsc{Set3}.

\textsc{Set1:} Of the 133 PCTL* requirements from \cite{grunske2008specification}, our probabilistic extension to \fret can express 108. The remaining 25 are unsupported due to the following: 
\new{
\begin{itemize}
    \item 22 use the bounded until operator or time intervals other than~[0,N], which classical \fret (and hence our probabilistic extension) does not support. These limitations are inherited in our extended \fret from classic \fret.
    \item two involve nested probabilities, which are not yet handled by the extension.
    \item one uses a probability structure ($P_{\geq 1}$ inside a $P_{\sim \text{bound}}$), referred to here as inverted probabilities, which also falls outside of the supported syntax.
\end{itemize}  
}

\textsc{Set2:} We expressed 175 of the 199 requirements from papers citing~\cite{grunske2008specification}. The 24 unsupported requirements include: bounded until (12), nested probabilities (2), time intervals other than~[0,N] (5), nested temporal operators (2) and inverted probabilities (3). 

\textsc{Set3:} We expressed 30 of the 38 QVBS requirements. The 8 unsupported requirements involved bounded until (7) or nested temporal operators (1), not supported in classical FRET.

To answer \textbf{RQ1}, probabilistic \fretish enabled us to express the vast majority (334/391, $\sim{}85\%$) of existing publicly available requirements. Notably, most of the unsupported cases (44 out of 57) involved bounded or nested temporal operators --- limitations inherited from classic \fretish. \new{These include several requirements that use time-bounded intervals other than [0\ldots N] for some $N\in\mathbb{N}$. For example, consider the following requirement from \textsc{Set1}: ``\textit{The probability of the queueing network becoming full between 0.5 and 2 time units is less than 0.1''}. This requirement would be formalised in PCTL* as $P<0.1[F[0.5, 2]\: \mathit{full}]$, but classical FRET does not support temporal bounds like $[0.5, 2]$ and this is inherited by our probabilistic extension.} Addressing these gaps is planned for future work.

\subsection{Probabilistic Requirements in an Industrial Scenario (RQ2)}
\label{sec:industrialcs}

To answer \textbf{RQ2}, we evaluated our extension on an industrial case study from the aviation domain, provided by RTRC. The case study focuses on a perception system \new{that produces a single sensor output based on inputs from} traditional and Learning-Enabled Component (LEC) sensors, used collaboratively across various aircraft maneuvers such as autonomous taxiing, take-off, and landing.

%In contrast to traditional sensors, LECs involve greater risk and unpredictability. % due to inherent aleatory uncertainty. 
%For example, a CNN-based runway detector may yield inaccurate deviation estimates, regardless of data availability, as the model must generalise rather than overfit because overfitting would compromise performance in unseen scenarios.

%built on a combination of readings from (i)~traditional sensors (e.g., GPS, altimeter, IMU existing in commercial aircraft~\cite{295117}) and (ii)~Learning-Enabled Component (LEC) sensors. These are jointly used in different aircraft manoeuvres (e.g., autonomous taxiing, take-off and landing). In contrast to traditional sensors, LECs involve greater risk and unpredictability due to inherent aleatory uncertainty. For example, a CNN-based runway detector may yield inaccurate deviation estimates, regardless of data availability, as the model must generalise rather than overfit—overfitting would compromise performance in unseen scenarios.

\begin{figure}[t]
% Figures editable here: https://docs.google.com/presentation/d/1bf8vJSO-sxdfRN-ZSQv5ELFUCUhMFMon/edit?slide=id.p2#slide=id.p2
    \centering
    \includegraphics[width=0.75\linewidth]{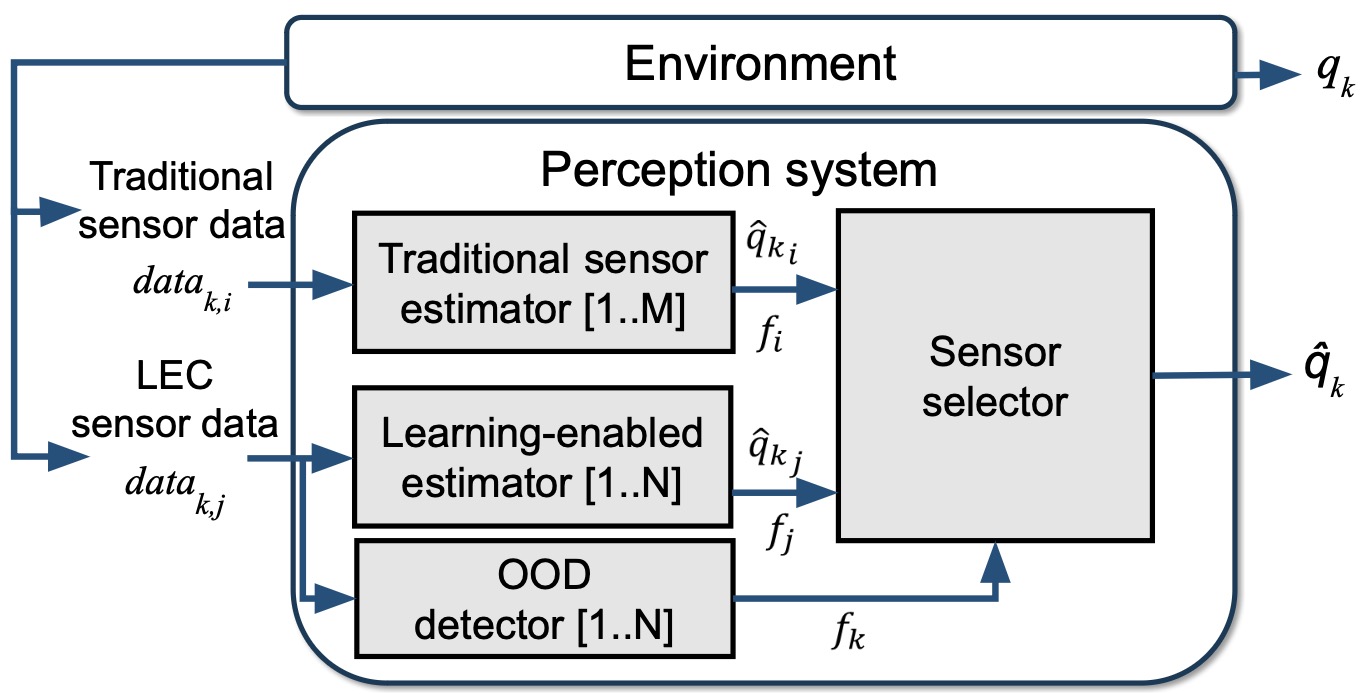}
      \caption{\new{Simplified RTRC-provisioned perception system overview. This system generates an estimate, $\hat{q}$, of $q$ at time $k$. Multiple sensors (traditional and LEC) collect data that are then processed by a series of estimators. The LEC components can be enhanced with out-of-distribution (OOD) detectors, which set a flag $f_k$ when an OOD value is predicted. A sensor selector combines these estimates to produce a single estimate, $\hat{q}_k$}.}
      \label{fig:simple_perception}
\end{figure}

Figure~\ref{fig:simple_perception} shows a simplified perception system architecture provided by RTRC. 
To estimate the value of the environment parameter $k$, the system processes various sensor inputs (e.g. camera images) and outputs a collection of parameter values $\hat{q}_k$. 
This passes through a series of (M+N) \textit{estimation pipelines} processing the raw incoming data (e.g., convolutional neural network (CNN) components for a LEC sensors). Each outputs an estimate $\hat{q}_{k,i|j}$ and has an integrated fault detection logic~\cite{bonfe2006fault} that outputs Boolean variables $f_{i|j}$ indicating a sensor fault. LEC components also include an out-of-distribution (OOD) detector that generate Boolean flags $f_k$ when sensor data are OOD. A sensor selection component then combines these estimates to produce a final estimate $\hat{q}_k$, which can be used by controllers or contingency management systems. Prediction accuracy is assessed by comparing the estimate, $\hat{q}_{k}$, against the actual ground truth, $q$.

\begin{table*}[!ht]
    \centering    
    \caption{\fretish and generated PCTL* formulas for our industrial case study. Complete list available in~\cite{FRETGithubX}.} 
\label{table:autonomy_examples_filtered}
    %\small
    \resizebox{\textwidth}{!}{%
    \begin{tabular}{p{1.6cm}L{4.2cm}L{5.2cm}L{2cm}p{7.8cm}}
    \toprule
         \textbf{ID} & \textbf{Natural-Language} &  \textbf{\fretish} & \textbf{Template Key} &
         \textbf{PCTL*}  \\\midrule
         
         \textbf{[P-001]} & \textit{The sensor selection component shall always satisfy accurate under ideal conditions.}  %(e.g. perfect weather, perfect lighting, perfectly marked airport, etc.)}
          & \conditionF{whenever idealConditions} \component{SensorSelection} shall \timing{immediately} \responseF{q\_hat = q}  & [null, holding, null, immediately] & \texttt{P>=1[(G (idealConditions => (P>=1[(q\_hat = q)])))]}
         \\ \midrule
        \textbf{[P-006]} & \textit{In auto-takeoff mode, whenever runway incursion occurs, the sensor selection output shall with probability > 0.99 at the next time point satisfy detection of runway incursion.} 
        & \scope{in auto\_takeoff\_mode} \conditionF{whenever q\_k} \component{SensorSelection} shall \probability{with probability $> 0.99$} \timing{at the next timepoint} \responseF{incursionDetected} 
        &[in, holding, bound, next]& \texttt{P>=1[((G ((! (((! auto\_takeoff\_mode) \& (X auto\_takeoff\_mode)))) | (X ((auto\_takeoff\_mode \& (X (! auto\_takeoff\_mode))) R (q\_k => (P>0.99[((auto\_takeoff\_mode \& (X (! auto\_takeoff\_mode))) | ((X incursionDetected) \& (! (auto\_takeoff\_mode \& (X (! auto\_takeoff\_mode))))))])))))) \& (auto\_takeoff\_mode => ((auto\_takeoff\_mode \& (X (! auto\_takeoff\_mode))) | ((auto\_takeoff\_mode \& (X (! auto\_takeoff\_mode))) R (q\_k => (P>0.99[((auto\_takeoff\_mode \& (X (! auto\_takeoff\_mode))) | ((X incursionDetected) \& (! (auto\_takeoff\_mode \& (X (! auto\_takeoff\_mode))))))]))))))]} \\
\hline 

   %       \textbf{[P-004]} & \textit{In auto-taxi mode, whenever the LECs disagree, the sensor selection output shall with probability > 0.99 within 5 ticks satisfy accurate } 
     %     & \scope{in auto\_taxi\_mode} \conditionF{whenever LEC\_disagreement} \component{SensorSelection} shall \probability{with probability $> 0.99$} \timing{within 5 ticks} \responseF{q\_hat = q} 
    %      & [in, holding, bound, within] 
    %      & \texttt{P>=1[((G ((! (((! auto\_taxi\_mode) \& (X auto\_taxi\_mode)))) | (X ((auto\_taxi\_mode \& (X (! auto\_taxi\_mode))) R (LEC\_disagreement => (P>0.99[(F<=5 (q\_hat = q))])))))) \& (auto\_taxi\_mode => ((auto\_taxi\_mode \& (X (! auto\_taxi\_mode))) | ((auto\_taxi\_mode \& (X (! auto\_taxi\_mode))) R (LEC\_disagreement => (P>0.99[(F[<=5] (q\_hat = q))]))))))]}\\  \midrule 

         \textbf{[P-007]} & \textit{after auto-land mode, the sensor selection output shall with probability > 0.99 eventually satisfy detection of correct runway exit} 
         & \scope{after auto\_land\_mode} \component{SensorSelection} shall \probability{with probability $> 0.99$} \timing{eventually} \responseF{detect\_correct\_exit} 
         & [after, null, bound, eventually] 
         & \texttt{P>=1[(((! (auto\_land\_mode \& (X (! auto\_land\_mode)))) U ((auto\_land\_mode \& (X (! auto\_land\_mode))) \& (X (P>0.99[(F detect\_correct\_exit)])))) | (G (! (auto\_land\_mode \& (X (! auto\_land\_mode))))))]} \\ \midrule

         \textbf{[P-012]} & \textit{Upon a runway incursion, the Runway Intrusion Detector, with probability greater than 0.99\%, detects the incursion before an unsafe separation distance is reached.} 
         & \conditionF{upon q\_k} \component{RunwayIntrusionDetector} shall \probability{with probability $> 0.9999$} \timing{before unsafe\_sep\_distance} \responseF{incursionDetected} 
         & [null, regular, bound, before] 
         & \texttt{P>=1[((G (((! q\_k) \& (X q\_k)) => (X (P>0.9999[(incursionDetected R (! unsafe\_sep\_distance))])))) \& (q\_k => (P>0.9999[( incursionDetected R (! unsafe\_sep\_distance))])))]} \\ \hline
         \textbf{[P-017]}& Whenever a runway incursion, RunwayIntrusionDetector, with probability greater than (99.99\%), detects an incursion before 10 time units. & \conditionF{whenever q\_k} \component{RunwayIntrusionDetector} shall \probability{with probability > 0.9999} \timing{within 10 ticks} \responseF{incursionDetected} & [null, holding, bound, within] & \texttt{P>=1[(G (q\_k => (P>0.9999[(F<=10 incursionDetected)])))]}\\\hline
         \textbf{[P-019]} & Upon accurate, the Runway Detector with probability greater than 99\% remains accurate for 10 time units. & \conditionF{upon accurate} \component{RunwayDetector} shall \probability{with probability > 0.99} \timing{for 10 ticks} \responseF{q\_hat =q} & [null, regular bound, for] & \texttt{P>=1[((G (((! accurate) \& (X accurate)) => (X (P>0.99[(G<=10 (q\_hat = q))])))) \& (accurate => (P>0.99[(G[<=10] (q\_hat = q))])))]}\\\bottomrule
    \end{tabular}}
    \vspace{-4mm}
\end{table*}

For this case study, we were able to express all of the 21 natural language requirements (\textsc{Set4}) that were provided by our industrial partner, RTRC. A subset is shown in \new{Table} \ref{table:autonomy_examples_filtered}. While the majority of these requirements were probabilistic, 2 of them were not (e.g.,  \textbf{[P-001]} in Table~\ref{table:autonomy_examples_filtered}). For these requirements, our \fret extension generated both LTL and PCTL$^*$, the latter was used for analysis by RTRC against their PRISM model.

Our results show that \fret can effectively support the specification of probabilistic requirements in this industrial case study. We validated that the formalizations matched the intended meaning of the natural language requirements 
through multiple iterations of discussions with our partner \new{(totalling approximately 30 hours over several months)}, who has expertise in formal logics.  \new{The feedback from our industrial partner was generally positive. }

\begin{quote}
\new{``\textit{Probabilistic FRET makes the task of formalizing requirements to PCTL* accessible to typical engineers without much formal methods background. It will also greatly accelerate the pace of formalizing requirements by engineers with formal methods backgrounds.}''}
\end{quote}

\noindent\new{They also outlined several avenues for future improvement.}

\begin{quote}
\new{``\textit{FRET like any language, requires a learning period before the user becomes very proficient at getting PCTL* that always matches the intention. The learning period however is short. There are certain parts of the language like conditions and modes which require a deeper grasp of the semantics that goes beyond the immediate natural intuition of a novice user. The simulator and the descriptions are very helpful in understanding some of those subtleties in the language.}''}
\end{quote}

\noindent\new{Exploring these usability aspects will form a centrepiece of our future projects with FRET. Specifically, we plan to extend classic FRET’s diagrammatic explanations to probabilistic requirements, making validation more intuitive and accessible to users with diverse backgrounds.}

%\marienote{mention differences between the sets, more modes in NL because mode is not always obvious from PCTL formulas}

%\marienote{requirements that are not probabilistic appeared in the industry use case, both LTL and PCTL* are generated for all formalizations. RTX wanted PCTL* for everything even tough they didn't have probabilities for analysis with PRISM.}

\subsection{Practical Utility of Extended \fretish (RQ3)}
To answer \textbf{RQ3}, we analyzed the structure of the probabilistic \fretish requirements in \textsc{Set1--4}. \new{Fig. \ref{fig:alltemplates} illustrates the spread and variety of the template keys that were used across each of the requirement sets. We focus on the most commonly used template keys in Fig. \ref{fig:consolidatedTemplates}. As captured in Table~\ref{tab:templatekeys}, }69.8\% of all 334 requirements across the four sets follow just five distinct patterns, \new{as shown in Fig. \ref{fig:alltemplates} and Fig. \ref{fig:consolidatedTemplates}.}
Interestingly, none of these patterns make use of the scope field, although scopes were used in 14/334 (4.2\%) of the requirements. Of these half come from the RTRC case study (7/14 requirements) representing 33\% of \textsc{Set4}. While prior literature \cite{grunske2008specification} rarely includes scopes (other than global), other large-scale studies using classic \fretish have shown more frequent use of scopes \cite{farrell2024fretting,mavridou2020LMCPS}, as does our industrial example. We discussed this observation that scopes did not appear often in the literature with our industry partner who remarked \textit{``that is surprising since modes are common in control and autonomous system software''.} This aligns with our experience, i.e., modes appear frequently in critical systems making FRET's capability to express mode-specific requirements particularly valuable in practice.

Among our contributions was also a new condition value for \fretish (holding condition, \conditionF{whenever} construct). This construct appeared in 34 out of 334 requirements (10.2\%), highlighting its practical utility.  \new{It is interesting to note the spread of the requirements across the template keys that is illustrated in Fig. \ref{fig:alltemplates} and in Fig. \ref{fig:consolidatedTemplates}, where a significant portion, mostly from \textsc{Set4}, fall into the ``Other'' category. This demonstrates that although some template keys were used more often than others, a great variety of the template keys were used in these requirement sets, i.e., 43~template keys in total.}

\begin{table}[t]
    \caption{Commonly Used Template Keys.}
    \label{tab:templatekeys}
    \centering
    \def\tabcolsep{4pt}
    \small
    \begin{tabular}{p{0.365\textwidth}llllp{0.19\textwidth}}
    \toprule
         \textbf{\textsc{Template Key}} & \textbf{\textsc{Set1}} & \textbf{\textsc{Set2}} & \textbf{\textsc{Set3}} & \textbf{\textsc{Set4}} & \textbf{\textsc{Total}} \\ 
         \midrule
         $[$null, null, bound, within$]$ & 40/108 & 31/175& 13/30& 0/21 & 84/334 (25.1\%)\\ 
         $[$null, null, bound, eventually$]$ & 19/108  &40/175 &11/30 & 0/21 & 70/334 (21\%) \\ %\midrule
         %\midrule
         $[$null, regular, bound, within$]$ & 18/108 & 20/175 & 0/30 & 1/21  & 39/334 (11.7\%) \\ %\midrule
         $[$null, null, bound, until$]$ & 8/108 & 12/175& 5/30 & 0/21& 25/334 (7.5\%) \\ %\midrule
         $[$null, holding, bound, within$]$ & 5/108 & 8/175& 0/30& 2/21& 15/334 (4.5\%)\\ \bottomrule
    \end{tabular}
    \vspace{-3mm}
\end{table}

\begin{figure}[t]
    \centering
    \includegraphics[width=\linewidth]{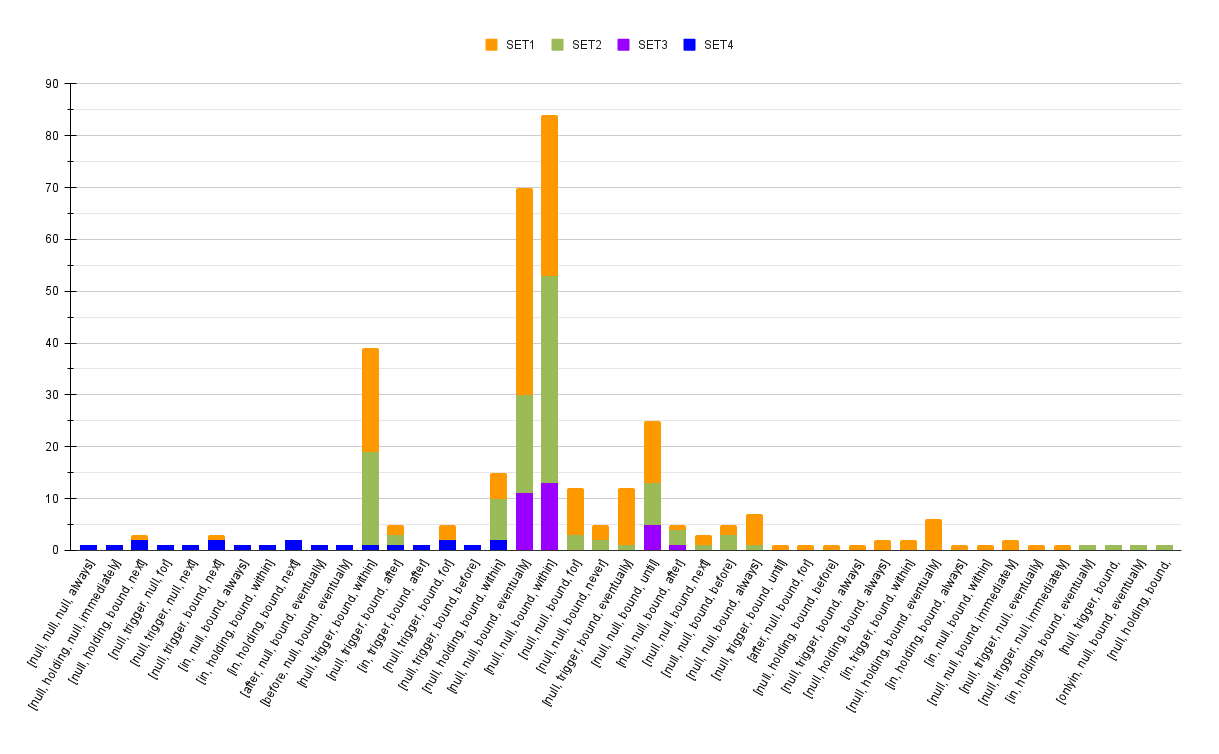}
    \caption{The total number of occurrences for each template key per requirement set.}
    \label{fig:alltemplates}
\end{figure}

%\begin{figure}
 %   \centering
%    \includegraphics[width=\linewidth]{images/templateoption2.png}
%    \caption{Option 2: Some moved to ``other''. Here, Other contains the templates with $\leq 2$ occurrences.}
%    \label{fig:option2}
%\end{figure}

\begin{figure}[t]
    \centering
    \includegraphics[width=\linewidth]{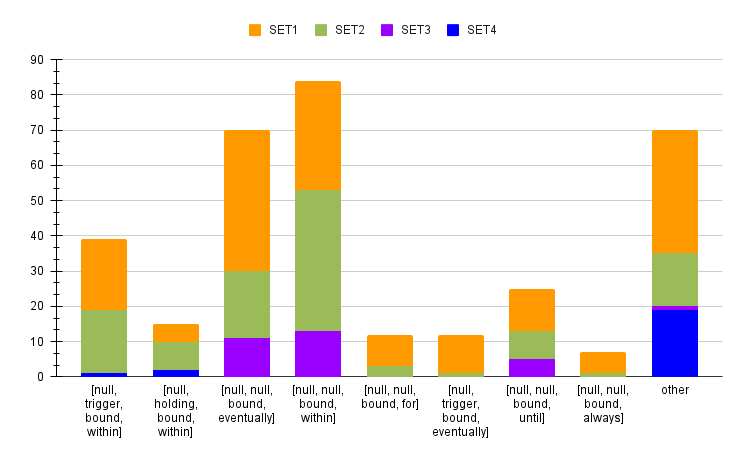}
    \caption{A consolidated view of the total number of occurrences for each template key per requirement set. Here, we moved template keys with $\leq 6$ occurrences to the ``Other'' category.}
    \label{fig:consolidatedTemplates}
\end{figure}

\section{Discussion}
\label{sec:discussion}

\subsection{\new{Challenges encountered}}

\new{
To extend \fret for probabilistic requirements, we had to carefully reason about design alternatives and address several non-trivial challenges.
\paragraph{Incorporating probabilistic semantics} We had to consider how to support probabilities. The addition of a new probability field felt natural but raised questions:  \textit{Where should this field reside within a requirement’s structure to allow nuanced probabilistic specifications while avoiding vacuous satisfaction?} Determining the precise placement required a detailed analysis of interactions with existing fields (see Subsection~\ref{sec:probabilityField}). 
\paragraph{Shifting from linear-time to branching-time reasoning} We had to reconsider \fret's reasoning approach. Classic \fret relies on Linear Temporal Logic (LTL), which interprets requirement satisfaction along a single execution trace. Probabilistic reasoning, however, naturally aligns with branching-time logic such as Probabilistic Computation Tree Logic (PCTL), which considers distributions over multiple possible futures. The challenge was compounded by the need to preserve the \fretish original design, and further amplified during integration with the pre-existing \fretish fields and in particular with the scope field. After extensive iteration, we addressed this challenge by introducing a two-step decomposition (see Algorithm~\ref{alg1}). While developing this decomposition, we uncovered subtle semantic nuances that could easily introduce errors. To ensure correctness, we developed a dedicated validation framework to verify that the resulting formal formulas faithfully captured the intended probabilistic semantics (see Section~\ref{sec:testing}).
\paragraph {Supporting probabilistic patterns observed in the research literature} We determined that many common probabilistic patterns described in literature could not be expressed using existing \fret constructs and the new probability field alone. Incorporating these without disrupting \fret’s language design was challenging, and required exploring multiple alternatives, including extending \fret fields and modifying template keys. Ultimately, we introduced a new \textit{condition type}, which proved essential for expressing a broad class of probabilistic patterns while preserving consistency with existing language constructs (see Section~\ref{sec:specification}).
\paragraph{Scale of template space} The number of FRETish template keys provides a quantitative measure of the language's expressive power, with each key representing a distinct, formally defined pattern. Consequently, a substantial increase in template keys corresponds to a meaningful expansion of \mbox{\fret's} capability to capture a broader range of requirements. Our extensions increase the total number of template keys supported by \fret from 160 to 560. This expanded set substantially increased the work required to extend \fret. Specifically, to finalize the formalization algorithm, we manually derived and reviewed 150+ individual formulas. Furthermore, the scale of the template space necessitated the development of an automated validation framework, while formalization issues uncovered during validation required analysis of more than 220 test results. 
\paragraph{Integration with the existing \fret toolchain}
Finally, it was essential to ensure that the probabilistic extension of \fret integrates seamlessly with the existing \fret features and preserves backward compatibility. This required that all core features, such as template selection, requirement validation, formal translation, and generation of formalizations, work correctly with probabilistic specifications.
}

\new{
\subsection{Structured Natural Language as a Mediator Between LLMs and Formal Specifications}
Large Language Models (LLMs) are beginning to be integrated into requirements engineering due to their advanced natural language processing capabilities, which offers immediate gains in efficiency. However, the inherent uncertain nature of LLMs and their tendency to hallucinate create a distinct tension with the field's need for precision. This is where tools such as \fret remain indispensable, particularly in safety-critical domains where reliance on LLMs is discouraged or strictly regulated. Even in less critical contexts, LLMs demonstrate higher effectiveness when used in conjunction with structured intermediate representations such as \fretish.
}

\new{
In one of our projects, we observed that translating ambiguous natural language requirements into \fretish before converting them into formal specifications significantly improved output accuracy compared to direct, one-step translation into temporal logic such as LTL. Our findings align with the work in~\cite{macedo2024intertrans}, which reports that employing structured intermediate representations can increase translation correctness by 18–43\% compared to direct translation approaches. 
Furthermore, using a structured language like \fretish not only enhances translation accuracy but also allows practitioners to leverage \fret's built-in explainability and formal analysis capabilities~\cite{giannakopoulou2020formal}. By operating at the \fretish level, engineers can obtain interpretable feedback on requirement consistency, detect ambiguities early, and systematically manage complex specification constraints. 
Taken together, these observations underscore the relevance and timeliness of our formally-defined structured natural language for probabilistic requirements. 
}
\section{Threats to validity}
\label{sec:threats}
%We summarize several threats to validity of this work.
\noindent\textit{Internal Threats.}  %Our validation models were generated randomly. %but a more guided generation might provide additional insight. 
Even though our validation approach is not exhaustive, we mitigate this threat by generating over 7000 models with diverse valuations of the \fretish fields and by providing a formal proof of correctness for Algorithm~\ref{sec:algorithm}.\medskip

\noindent\textit{External Threats}. We assembled four distinct requirement sets to evaluate our research questions, with each set chosen to offer complementary perspectives and coverage. While carefully constructed, certain limitations may influence the generalizability of our findings. \textsc{Set1} was derived from a seminal 2008 publication \cite{grunske2008specification}. Although still highly influential, it may not capture more recent developments.  \textsc{Set2} extends this by incorporating requirements from all papers citing~\cite{grunske2008specification}, reflecting the evolution of ideas within an active research community. A broader literature review could further diversify the set. \textsc{Set3} was sourced from an established benchmark, a standard practice in software engineering, though it may omit rare or emerging patterns. \textsc{Set4} was developed in collaboration with an industrial partner, ensuring practical relevance. Future work could benefit from engaging additional partners from different domains to further enhance coverage.

We collected a total of 334 probabilistic requirements. While we have made efforts to ensure that this corpus is both diverse and representative, it may not yet support broad generalization. Nonetheless, it stands as a concrete contribution of our work and establishes a basis for a growing benchmark of probabilistic requirements for future research and evaluation.

\section{Conclusion}
\label{sec:conclude}
As critical systems grow in complexity, incorporating uncertain sensor data and unpredictable machine learning components, it is becoming clear that software requirements must evolve to capture uncertainty. This is especially necessary in domains where rigorous formal development, verification and validation processes must be followed. 

To this end, we introduce a novel automated framework for formalizing probabilistic requirements from structured natural language, implemented as an extension to \fret. Our contributions include a compositional algorithm supported by a dedicated validation framework and a correctness proof, along with a comprehensive evaluation using both literature-based and industrial requirements from diverse domains. Notably, our tool was expressive enough for the majority of requirements; the few exceptions stemmed from limitations of classic \fret or involved certain forms of probabilistic nesting. These findings validate the utility of our approach and chart a path toward supporting other logics such as RPCTL.

%===============================================================
% Bibliography
%===============================================================

\subsection*{Acknowledgements}
A. Mavridou is supported by NASA Contract No. 80ARC020D001. M. Farrell is supported by a Royal Academy of Engineering Research Fellowship. G. Vázquez performed part of this work during her internship with KBR Inc. at NASA Ames Research Center. T. Pressburger is supported by NASA’s System-Wide Safety project in the Airspace Operations and Safety Program. The NASA University Leadership initiative (grant \#80NSSC20M0163) provided funds to assist T.E. Wang with this research.
R.~Calinescu is supported by the UK Advanced Research and Invention Agency's Safeguarded AI programme, and the York Centre for Assuring Autonomy. M. Fisher is supported by the Centre for Robotic Autonomy in Demanding and Long Lasting Environments (CRADLE) under EPSRC grant EP/X02489X/1 and by the Royal Academy of Engineering under the Chairs in Emerging Technology scheme. 

\bibliographystyle{elsarticle-num}
\bibliography{refs} 

\end{document}